\documentclass[sn-mathphys,Numbered]{sn-jnl}

\usepackage{graphicx}
\usepackage{multirow}
\usepackage{amsmath,amssymb,amsfonts}
\usepackage{amsthm}
\usepackage{mathrsfs}
\usepackage[title]{appendix}
\usepackage{xcolor}
\usepackage{textcomp}
\usepackage{manyfoot}
\usepackage{booktabs}
\usepackage{listings}
\usepackage{tcolorbox}

\usepackage{algorithm}
\usepackage{algorithmic}
\floatname{algorithm}{Method}

\usepackage[english]{babel}
\usepackage[latin1]{inputenc}
\usepackage{enumerate}
\usepackage{color}
\usepackage[T1]{fontenc}
\usepackage{epstopdf}
\usepackage{subfigure}
\usepackage{dsfont}
\usepackage[T1]{fontenc}
\usepackage{amsmath}
\usepackage{mathtools}
\usepackage{amsthm}
\usepackage{amstext}
\usepackage{amssymb}
\usepackage{mathrsfs}
\usepackage{mathtools}
\usepackage{tikz}
\usetikzlibrary{arrows,positioning}
\usepackage{pifont}

\def\D{{\mathcal D}}

\def\H{{\mathcal H}}

\def\C{{\mathcal C}}

\def\P{{\mathcal P}}

\def\M{{\mathcal M}}
\def\S{{\mathcal S}}
\def\cE{\mathbb{E}}

\def\bpi{{\boldsymbol \pi}}
\def\bgamma{{\boldsymbol \gamma}}

\def\ent{{\mathsf{ent}}}
\def\Mar{{\mathsf{Mar}}}
\def\Inv{{\mathsf{Inv}}}

\def\sTr{{\mathsf{Tr}}}

\def\sX{{\mathsf X}}

\def\sA{{\mathsf A}}

\def\sE{{\mathsf E}}

\def\cL{\mathcal{L}}
\def\R{\mathbb{R}}

\def\N{{\mathcal N}}
\def\ker{\mathbb{\rm ker}}

\def\Id{\mathbb{Id}}

\def\dim{\mathop{\rm dim}}
\def\diag{\mathop{\rm diag}}
\def\tr{\mathbb{\rm Tr}}

\def\T{\mathbb{\rm T}}
\def\rR{\mathbb{\rm R}}
\def\L{\mathbb{\rm L}}

\def\cC{\mathop{\rm C}}

\def\argmin{\mathop{\rm arg\, min}}

\def\range{\mathbb{\rm Range}}

\theoremstyle{thmstyleone}
\newtheorem{theorem}{Theorem}

\newtheorem{proposition}{Proposition}

\theoremstyle{thmstyletwo}

\theoremstyle{thmstylethree}
\newtheorem{definition}{Definition}
\newtheorem{assumption}{Assumption}
\newtheorem{remark}{Remark}

\begin{document}
\title[Article Title]{Quantum Markov Decision Processes: Dynamic and Semi-Definite Programs for Optimal Solutions}

\author*[1]{\fnm{Naci} \sur{Saldi}}\email{naci.saldi@bilkent.edu.tr}

\author[2]{\fnm{Sina} \sur{Sanjari}}\email{sanjari@rmc.ca}

\author[3]{\fnm{Serdar} \sur{Y\"uksel}}\email{yuksel@queensu.ca}

\affil*[1]{\orgdiv{Department of Mathematics}, \orgname{Bilkent University}, \orgaddress{\street{\c{C}ankaya}, \city{Ankara}, \country{T\"{u}rkiye}}}

\affil[2]{\orgdiv{Department of Mathematics and Computer Science}, \orgname{Royal Military College},  \state{Ontario}, \country{Canada}}

\affil[3]{\orgdiv{Department of Mathematics and Statistics}, \orgname{Queen's University}, \orgaddress{\city{Kingston}, \state{Ontario}, \country{Canada}}}

\abstract{In this paper, building on the formulation of quantum Markov decision processes (q-MDPs) presented in our previous work [{\sc N.~Saldi, S.~Sanjari, and S.~Y\"{u}ksel}, {\em Quantum Markov Decision Processes: General Theory, Approximations, and Classes of Policies}, SIAM Journal on Control and Optimization, 2024], our focus shifts to the development of semi-definite programming approaches for optimal policies and value functions of both open-loop and classical-state-preserving closed-loop  policies. First, by using the duality between the dynamic programming and the semi-definite programming formulations of any q-MDP with open-loop policies, we  establish that the optimal value function is linear and there exists a stationary optimal policy among open-loop policies. Then, using these results, we establish a method for computing an approximately optimal value function and formulate computation of optimal stationary open-loop policy as a bi-linear program. Next, we turn our attention to classical-state-preserving closed-loop policies. Dynamic programming and semi-definite programming formulations for classical-state-preserving closed-loop  policies are established, where duality of these two formulations similarly enables us to prove that the optimal policy is linear and there exists an optimal stationary classical-state-preserving closed-loop policy. Then, similar to the open-loop case, we establish a method for computing the optimal value function and pose computation of optimal stationary classical-state-preserving closed-loop policies as a bi-linear program.}


\maketitle

\section{Introduction}\label{sec0}

In this continuation of our work on quantum Markov decision processes (q-MDPs), we build on the novel mathematical framework introduced in our previous paper \cite{SaSaYu24}. Our focus here shifts from theoretical foundations to the development of semi-definite programming techniques aimed at deriving optimal policies and value functions for both open-loop and classical-state-preserving closed-loop scenarios. While our earlier work unified classical MDP theory with recent advancements in quantum MDPs, creating a coherent framework, this paper extends those theoretical insights into methodologies for achieving optimal solutions. Here, we aim to adapt classical techniques such as dynamic and linear programming to the quantum domain.

Classical MDPs, along with their partially observable and decentralized counterparts, have long been recognized as versatile models with significant relevance across a wide range of fields, including engineering, economics, and natural and applied sciences \cite{HeLa96}. The primary objective in MDP theory is to identify an optimal control policy, typically aimed at minimizing a cost function. Common cost criteria include the discounted cost, average cost, or finite horizon cost. In this paper, we concentrate on the discounted cost scenario, whose analysis can be done via the two well-established approaches: dynamic programming and linear programming. These two methods are intricately connected, with dynamic programming often considered as the dual of linear programming \cite{HeLa96}.

Despite the extensive research on classical MDPs, a comprehensive study of their quantum counterparts has been lacking, even though several recent studies have explored various aspects of quantum MDPs \cite{BaBaAa14,YiYi18,YiFeYi21,SaYu22}. The most well-known early attempt to extend the classical MDP model into the quantum domain is found in \cite{BaBaAa14}, where the authors introduce quantum observable Markov decision processes (QOMDPs) as a quantum analogue to partially observable MDPs (POMDPs). In a QOMDP, the state is represented by a density operator, and the set of actions consists of quantum channels (super-operators) that can be applied to these states. Each chosen super-operator is decomposed into Kraus operators, with each Kraus operator representing an observation. The next state is determined based on the outcome of the chosen super-operator, and the cost is incurred based on the chosen action and the current state. The primary focus of \cite{BaBaAa14} is on comparing complexity-related problems between classical and quantum settings. 

The primary objective of our work is to develop quantum analogs of dynamic programming and linear programming methods, with a particular focus on the latter. Unlike \cite{BaBaAa14}, which emphasizes complexity-related issues, our approach prioritizes methodological development. These quantum analogues will serve as essential tools for achieving optimal policies and value functions within q-MDPs, applicable to both open-loop and classical-state-preserving closed-loop policies.

Another related study is \cite{YiYi18}, where the authors propose a different quantum version of MDPs. In their model, the state evolves with respect to selected indivisible super-operators, while at certain times, the controller is allowed to take measurements on the system state via divisible super-operators. This approach integrates both indivisible and divisible quantum operations to model state evolution, with again a focus on addressing complexity-related issues.

In \cite{YiFeYi21}, the authors present a quantum version of MDPs that is particularly relevant to our work, as it also aims to compute optimal policies and value functions. In this model, the state is represented as a density operator on a finite-dimensional Hilbert space, and an action is selected from a finite set. The state then transitions to an intermediate state via a super-operator parameterized by the chosen action. Subsequently, a measurement is conducted on this intermediate state to determine both the reward function and the next state.

Given the similarities among these models, we collectively refer to them as QOMDPs. It is important to note, however, that there are key distinctions between QOMDPs and the q-MDP model we propose. For example, QOMDPs can be formulated as classical MDPs with an uncountable compact state space (i.e., the set of density operators on a finite-dimensional Hilbert space) and a finite action space. This formulation allows for the application of classical methodologies, such as dynamic programming and linear programming, possibly with a discrete approximation of the state space. However, within the q-MDP context, the direct application of classical techniques is not feasible; instead, these methods must be adapted to suit the quantum setup, particularly in the case of linear programming.

Moreover, since QOMDPs can be expressed as classical MDPs with uncountable compact state spaces, they can be further transformed into deterministic models by lifting the state and action spaces to the set of probability measures over the state and state-action spaces. This deterministic model can then be further transformed into a q-MDP, as classical deterministic MDPs over probability distributions can always be represented as a special case of our quantum model. Therefore, we can view QOMDPs as specialized instances of q-MDPs. However, a challenge arises in that the corresponding Hilbert space for the state must be infinite-dimensional-a scope not covered in this work. We address this complexity by approximating the classical version of QOMDPs with finite-state MDPs, whose convergence to the underlying MDP can be established using the results in \cite{SaYuLi17}. We then embed the finite-state approximate MDPs into the q-MDP formulation with finite-dimensional Hilbert spaces, enabling us to model QOMDPs approximately as a special case of q-MDPs within finite-dimensional Hilbert spaces.

In \cite{SaYu22}, the authors discuss the quantum generalization of static decentralized control systems. Since this model is static, it lacks state dynamics and thus does not involve dynamical difference equations, unlike the MDP problem. The primary goal in this generalization is to introduce additional correlations among decentralized decision-makers by allowing them to share an entangled state and perform independent quantum measurements on their respective parts of the entangled state. This approach creates extra correlations that can enhance system performance. This generalization of static decentralized control systems can be considered semi-quantum, as everything in the model remains unchanged except for the conditional distribution of each agent's action given its observation, which is determined by the outcome of a quantum measurement. Nevertheless, our q-MDP framework can be extended to such multi-agent and decentralized dynamic quantum formulations by integrating the model presented in our work with the space of decentralized control policies studied in \cite{SaYu22}.

Furthermore, there have been numerous foundational contributions in the control theory literature concerning the optimal control of quantum systems, primarily focusing on the continuous-time setting \cite{AlTi12,Dal08,JaKo18,DoPe22,BeKoSc24}. As in classical optimal control theory, while there are conceptual similarities between approaches developed for continuous-time and discrete-time systems, significant technical distinctions exist between the two frameworks.

There are also conceptual differences between the methodologies employed in the control theory literature and our approach to modeling quantum control systems. In conventional control theory, a quantum control system is typically considered a direct extension of deterministic control models, expressed through controlled differential equations. In the quantum domain, the governing differential equation for quantum dynamics is the controlled Schr\"{o}dinger equation. 

If the density operator, governed by the Schr\"{o}dinger equation, solely represents the state of a quantum system, it describes a closed quantum control system, implying no interaction with the environment. This is because the evolution of the state is controlled by a semi-group of unitary operators. However, when considering interactions between the quantum system and its environment (open quantum systems), the Schr\"{o}dinger equation is defined over the combined system-environment space. A reduced description of the state dynamics for the quantum system of interest can be obtained through the partial trace operation, which typically results in non-Markovian dynamics \cite{AlTi12}. This is analogous to classical partially observable MDPs, where the joint state of two subsystems forming a Markovian process does not imply that an individual state, which may be viewed as a measurement of the hidden Markov model, is itself Markovian. Non-Markovianity arises due to the system's interaction with an environment, leading to memory effects in the reduced system. Although there has been some works on controlling non-Markovian dynamical systems, most control-oriented research has concentrated on the Markovian case. This focus is largely due to the significant advantage of Markovian systems, which yield simpler dynamical equations \cite{AlTi12}. These simpler equations allow for the application of numerous well-established methods from classical control theory, making analysis and control design more tractable. Indeed, under appropriate assumptions, non-Markovian dynamics can be approximated by a Markovian structure, leading to a Markovian master equation in the Lindblad form \cite{Lin76}. For example, in the weak coupling limit, one can derive the Markovian master equation, where the joint state of the quantum system and the environment is approximated by a Markovian structure, with the environment state remaining invariant under the Hamiltonian action \cite{BrPe07}. In our formulation, as we do not begin with a Hamiltonian description of the quantum system, the model avoids the memory effect that typically leads to non-Markovian dynamics.

In both closed quantum systems and Markovian approximations of open quantum systems, quantum control models can be treated similarly to classical control models. Consequently, much of the research in quantum control has focused on classical control issues, such as stability \cite{TiVi09}, controllability \cite{Alt03,Alt04}, identification \cite{KoRaWa03b}. There has also been significant work on robust control of quantum systems \cite{KoGrMaBrRa13,KoBhRa22,KoBhRa23} and adaptive control approaches \cite{KoRaGr13, KoRaWa03}. Several studies have also explored the optimal control of quantum systems \cite{PeDaRa88,ZhRa98,Dal08}. Since quantum control systems are fundamentally viewed as extensions of classical control models, established techniques from classical optimal control, such as the Pontryagin maximum principle \cite{BoSiSu221} and the Hamilton-Jacobi-Bellman formalism \cite{GoBeSm05}, can be adapted to quantum systems. However, solving optimal control problems for non-Markovian dynamics in open quantum systems remains a significant challenge \cite{BiGo12}. This is because most classical control methods rely on the Markovian assumption for state dynamics, which does not generally hold for open quantum systems. This limitation represents a major difficulty in the quantum control literature, particularly for continuous-time open systems. In addition, continuous-time quantum control has been explored in the context of quantum information processing, particularly in addressing error correction challenges \cite{HoZhKoBrRa16}.

Unlike the continuous-time approach, we do not directly translate the controlled stochastic difference equation that governs the state of the classical system into an equation on the set of density operators in discrete time. Instead, we first elevate our classical stochastic control system to the space of probability measures over the state and state-action spaces. The dynamics of this expanded system can be fully described using classical channels, thereby eliminating the need for stochastic difference equations. Given that density operators and quantum channels are the quantum counterparts of probability measures and classical channels, we seamlessly derive the quantum version of this extended model. As a result, our method avoids the issue of non-Markovian dynamics and is applicable to both closed and open systems without the need for separate descriptions and interactions with an unknown environment.

In this paper, we build on the formulation of q-MDPs presented in our previous work \cite{SaSaYu24} to develop semi-definite programming techniques for deriving optimal policies and value functions in both open-loop and classical-state-preserving closed-loop scenarios. In classical Markov decision processes (MDPs), dynamic programming and linear programming are well-established methods that provide distinct yet intricately connected approaches to achieving optimal policies and value functions, with each method serving as the dual of the other \cite{HeLa96}. By extending these classical techniques into the quantum domain, we leverage this duality to establish methodologies for deriving optimal open-loop and classical-state-preserving closed-loop policies.

\subsection{Contributions}
\begin{itemize}
\item[\ding{43}] In Section~\ref{sec2}, we develop methodologies for deriving optimal policies and optimal value functions for q-MDPs with open-loop policies. To this end, we establish the dynamic programming formulation in Section~\ref{sub2sec2} (see Theorem~\ref{dynamic-qmdp}) and the semi-definite programming formulation in Section~\ref{sub3sec2} for open-loop policies. Then, by using the duality between the dynamic programming and semi-definite formulations, we prove the existence of the optimal stationary open-loop policy (see Theorem~\ref{quantum-stationary}), establish method that gives us approximate value function (see Proposition~\ref{open-opt-cost-prop}), and formulate computation of optimal stationary open-loop policy as a bi-linear program (see Proposition~\ref{open-optimal-policy-prop_comp}). 
\item[\ding{43}] In Section~\ref{sec3}, our focus shifts to the study of classical-state-preserving closed-loop  policies. In this section, we first obtain the dynamic programming principle (see Theorem~\ref{dynamic-qwmdp}) and establish the semi-definite programming formulation. Then, by using the duality of these two formulations, we establish the existence of the optimal stationary classical-state-preserving closed-loop  policy (see Theorem~\ref{closed-quantum-stationary}), establish method that gives us the value function (see Proposition~\ref{closed-opt-cost-prop}), and formulate computation of optimal stationary classical-state-preserving closed-loop policy as a bi-linear program (see Proposition~\ref{closed-optimal-policy-prop_comp}). 
\item[\ding{43}] In Section~\ref{sec1}, we revisit the deterministic reduction of classical MDPs, which is a well-known result in the literature, as documented in Section 9.2 of \cite{BeSh78}. Following the approach outlined in \cite[Chapter 6]{HeLa96}, we then present a linear programming formulation for this deterministic reduction. 
\end{itemize}

Figure~\ref{Hier} presents the hierarchy of policies defined in this paper, offering a structured overview of the policy classes and their relationships.

\begin{figure}[h]
\centering
\tiny
\begin{tikzpicture}[scale=0.75]

\draw[thick] (-8,-6) rectangle (8,6);

\draw[thick] (0,0) ellipse (5.5 and 5.5);
\node at (-5.5,5.5) {\textbf{Quantum Policies (QPs)}};
\node at (-0,5) {\textbf{Markov QPs}};

\draw[thick] (0,0) ellipse (4.5 and 4.5);
\node at (0,3) {\textbf{Closed-loop QPs}};
\node at (0,3.5) {\textbf{Clasical-State-Preserving}};

\draw[thick] (-1.5,-0.5) ellipse (2.6 and 2.6);
\node at (-2.5,-0.5) {\textbf{Classical QPs}};

\draw[thick] (1.5,-0.5) ellipse (2.6 and 2.6);
\node at (2.5,-0.5) {\textbf{Open-loop QPs}};

\end{tikzpicture}
\normalsize
\caption{Hierarchy of policies: (i) (History dependent) quantum policies utilize all historical information (and such policies may not be physically realizable, See Remark~\ref{realizability}). (ii) Markov quantum policies rely solely on current state information without any additional constraints (and may not be physically realizable). (iii) Classical-state-preserving closed-loop quantum policies use the current state information but with a relaxed invertibility condition on the quantum channel, mirroring the classical invertibility condition. (iv) Open-loop quantum policies use current state information but impose a stricter invertibility condition on the quantum channel, again reflecting the classical case. (v) Classical quantum policies represent an embedding of classical policies within a quantum framework. We refer the reader to Remark~\ref{realizability} on physical realizability.}
\label{Hier}
\end{figure}

Throughout the rest of this paper, we aim to devise methodologies using dynamic programming and semi-definite programming methods to achieve optimal value and policy for open-loop policies and classical-state-preserving closed-loop  policies, respectively. Initially, we  examine the open-loop case.

For computational purposes, throughout the remainder of this paper, we confine our exploration solely to q-MDPs with finite dimensional Hilbert spaces $\H_{\sX}$ and $\H_{\sA}$. Additionally, we limit our focus to classical MDPs with a finite number of states and actions (i.e., $|\sX| < \infty$ and $|\sA| < \infty$). Consequently, we typically let $\dim(\H_{\sX}) = |\sX|$ and $\dim(\H_{\sA}) = |\sA|$.

\section{Quantum Markov Decision Processes}

In this section, we present the q-MDPs formulation introduced in our previous work \cite{SaSaYu24}. To ensure the discussion is self-contained, we first revisit the motivation for this model as described in \cite{SaSaYu24}. We begin by outlining the formulation of classical MDPs and their deterministic reduction, a well-established process that involves extending the state and action spaces to the set of probability measures on the state and state-action spaces. This reduction, detailed in Section 9.2 of \cite{BeSh78}, provides a critical foundation for our q-MDPs formulation. Using this deterministic MDP framework as a base, we then introduce the concept of q-MDPs. Along the way, we also provide a brief introduction to quantum systems. This section serves as a summary of our previous work \cite{SaSaYu24}, included here to maintain a self-contained presentation.

\subsection{Deterministic Reduction of Classical MDPs}\label{sec2new-1}

A discrete-time classical Markov Decision Process (MDP) is characterized by a five-tuple $\bigl( \sX, \sA, p, c \bigr)$, where $\sX$ and $\sA$ are finite sets representing the \emph{state} and \emph{action} spaces, respectively. The \emph{stochastic kernel} $p(\,\cdot\,|x,a)$ specifies the \emph{transition probability} of moving to the next state given the current state-action pair $(x,a)$. The \emph{one-stage cost} function $c$ maps from $\sX \times \sA$ to $\R$. A policy is a sequence $\bpi=\{\pi_{t}\}$ of stochastic kernels from $\sX$ to $\sA$, with the set of all possible policies denoted by $\Pi$. Let $\Phi$ represent the set of stochastic kernels $\varphi$ from $\sX$ to $\sA$. A \emph{stationary} policy is a constant sequence $\bpi=\{\pi_{t}\}$ where $\pi_{t}=\pi$ for all $t$, for some fixed $\pi \in \Phi$. The set of stationary policies is therefore equivalent to $\Phi$. Given an initial distribution $\mu_0$ on $\sX$ and a policy $\bpi$, the evolution of states and actions is described by
\begin{align}
x_0 \sim \mu_0, \nonumber \,\,\,\,
x_t \sim p(\,\cdot\,|x_{t-1},a_{t-1}), \text{ } t\geq 1, \,\,
a_t \sim \pi_t(\,\cdot\,|x_t), \text{ } t\geq0. \nonumber
\end{align}
In MDPs, the usual objective is to minimize cost measures such as finite horizon cost, discounted cost, and average (or ergodic) cost criteria, denoted as $J(\bpi,\mu_0)$. The optimal value of the MDP is defined as $J^*(\mu_0) \coloneqq \inf_{\bpi \in \Pi} J(\bpi,\mu_0)$. A policy $\bpi^{*}$ is considered optimal for $\mu_0$ if $J(\bpi^{*},\mu_0) = J^*(\mu_0)$. 

By employing a method similar to the one outlined in Section 9.2 of \cite{BeSh78}, any MDP can be transformed into an equivalent deterministic model.

To define the equivalent MDP, for each $\mu \in \P(\sX)$, we introduce
$\sA(\mu) \coloneqq \left\{\nu \in \P(\sX \times \sA) \mid v(\,\cdot\,\times\sA) = \mu(\,\cdot\,)\right\}.$
The deterministic version of an MDP, referred to as d-MDP, is described by a tuple
$
\left( \P(\sX), \P(\sX \times \sA), P, C \right),
$
where $\P(\sX)$ represents the state space and $\P(\sX \times \sA)$ represents the action space. The state transition function $P: \P(\sX \times \sA) \rightarrow \P(\sX)$ is defined by
\begin{align}
P(\nu)(\,\cdot\,) = \int_{\sX \times \sA} p(\,\cdot\,|x,a) \, \nu(dx,da). \label{eq1}
\end{align}
The \emph{one-stage cost} function $C$ is a linear map from $\P(\sX \times \sA)$ to $\R$ and is given by
\begin{align}
C(\nu) = \int_{\sX \times \sA} c(x,a) \, \nu(dx,da) \eqqcolon \langle c,\nu \rangle. \label{eq2}
\end{align}
In this setup, a Markov policy is a sequence of classical channels $\bgamma=\{\gamma_{t}\}$ from $\sX$ to $\sX \times \sA$, such that for all $t$,
\begin{align}
\gamma_t(\mu) \in \sA(\mu), \,\, \forall \,\mu \in \P(\sX). \label{inv_classical}
\end{align}
The set of all policies is denoted by $\Gamma$. A \emph{stationary} policy is a constant sequence of channels $\bgamma=\{\gamma\}$ from $\sX$ to $\sX \times \sA$, where $\gamma(\mu) \in \sA(\mu)$ for all $\mu \in \P(\sX)$.

\subsection{Quantum Systems}

Before introducing the formulation of q-MDPs, it is beneficial to review some fundamental concepts from quantum mechanics and relate them to classical systems. For detailed terminology and definitions, we refer the reader to \cite{NiCh02,Wil13,Wat18,Par92}.

A quantum system is described by a finite-dimensional complex Hilbert space $\H$, which represents the state space of the system. In quantum mechanics, we use Dirac notation to represent elements of $\H$ as ket-vectors, denoted $|\psi \rangle$. The corresponding conjugate transpose of a ket-vector $|\psi \rangle$ is called a bra-vector, denoted $\langle \psi |$. Using this notation, the inner product between two vectors can be expressed as $\langle \psi | \xi \rangle$.

According to the first postulate of quantum mechanics, the complete description of a quantum state is given by a density operator $\rho$. This operator is positive semi-definite and has a unit trace, meaning
$
\langle \psi | \rho | \psi \rangle \geq 0 \text{ for all } \psi \in \H \text{ and } \tr(\rho) \coloneqq \sum_{i=1}^n \rho_{ii} = 1,
$
where $n$ is the dimension of $\H$, and $\{\rho_{ii}\}$ are the diagonal elements of $\rho$ in any orthonormal basis of $\H$. If $\rho$ has rank one, i.e., $\rho = |\psi \rangle \langle \psi|$, where $|\psi \rangle$ is a unit vector, then $\rho$ is termed a pure state. Otherwise, it is referred to as a mixed state. The set of all density operators on $\H$ is denoted by $\D(\H)$, which is known to be a convex set with its extreme points being the pure states \cite[page 46]{Wat18}.

A common example of a density operator is one that represents a probability distribution $\mu$ over the set $[n] \coloneqq \{1,2,\ldots,n\}$, where $n$ is the dimension of $\H$. Specifically, if we represent any density operator as a matrix with respect to some orthonormal basis, then $\rho \coloneqq \diag\{\mu(i)\}$ defines a density operator since $\langle \psi | \rho | \psi \rangle \geq 0 \text{ for all } \psi \in \H$ due to $\mu(i) \geq 0$ for all $i$, and $\tr(\rho) = 1$ since $\sum_{i=1}^n \mu(i) = 1$. Thus, any probability distribution on $[n]$ can be interpreted as a density operator on $\H$, which are known as classical states.

Using above analogy, the expectation operation can be performed in the quantum framework as follows. Consider a random variable \( X \) defined on the probability space \(([n], \mu)\). We can represent this random variable as \(\phi \coloneqq \text{diag}\{X(i)\}\). The expectation of \( X \) with respect to the probability distribution \(\mu\) can then be computed in the quantum context as:
$
\cE[X] = \langle X, \mu \rangle \coloneqq \sum_{i} X(i) \mu(i) = \tr(\rho \phi) \coloneqq \langle \rho, \phi \rangle,
$
where the first inner product is the Euclidean inner product, and the second is the Hilbert-Schmidt inner product. Hence, in quantum mechanics, expectation is computed using the trace operation.

According to the second postulate of quantum mechanics, the evolution of an isolated quantum system, or a closed quantum system, is governed by a unitary transformation. Specifically, the state $\rho^+$ of a closed quantum system evolves from $\rho$ via a unitary operator $U$ (with $U U^* = U^* U = \Id$):
$
\rho^+ = U \rho U^*.
$
For example, the bit-flip operator is a unitary transformation for a two-dimensional quantum system:
$$
\begin{pmatrix}
0 & 1 \\
1 & 0
\end{pmatrix}
$$
This operator flips the qubit state $|0\rangle$ to $|1\rangle$ and vice versa.

In practice, it is difficult to prevent a quantum system from interacting with its environment, which leads to open quantum systems. The evolution of such systems can be described as follows. For two quantum systems with Hilbert spaces $\H_1$ and $\H_2$, their combined system is represented by the tensor product Hilbert space $\H_1 \otimes \H_2$. Let $\H_e$ denote the Hilbert space associated with the environment, and $\rho_e$ be the state of the environment. The combined state of the system and environment is then $\rho \otimes \rho_e$. Since this combined system $\H \otimes \H_e$ is treated as a closed quantum system, its evolution is governed by a unitary operator $U_c$:
$
\xi^+ = U_c (\rho \otimes \rho_e) U_c^*,
$
where $\xi^+ \in \D(\H \otimes \H_e)$. Note that $\xi^+$ may not remain in a simple product form. To recover the next state of our quantum system $\H$ from the combined state $\xi^+$, we apply the partial trace operation.

The partial trace operation, denoted by $\tr_{E}: \D(\H \otimes \H_e) \rightarrow \D(\H)$, is defined as $\tr_E \coloneqq \Id \otimes \tr$. When applied to product states like $\rho \otimes \rho_e$, it acts as:
\[
\tr_E(\rho \otimes \rho_e) = \rho \, \tr(\rho_e) = \rho \quad \text{since} \quad \tr(\rho_e) = 1.
\]
The partial trace operation is analogous to marginalization in classical probability. For instance, if $\nu$ is a probability distribution on $[n] \times [m]$, where $m$ is the dimension of $\H_e$, the marginalization on $[m]$ is:
\[
\Mar_{E}(\nu) = \nu(\cdot \times [m]) \eqqcolon \mu(\cdot),
\]
where $\mu$ is the marginal distribution of $\nu$ on $[n]$. If $\xi$ and $\rho$ are classical density operators corresponding to $\nu$ and $\mu$, respectively, then:
\[
\rho = \tr_E(\xi).
\]
Thus, the marginalization operation $\Mar_{E}$ in classical probability is implemented in the quantum domain via the partial trace operation $\tr_E$.

To determine the next state of our quantum system from the combined system, we use the partial trace operation:
\[
\rho^+ = \tr_E(U_c (\rho \otimes \rho_e) U_c^*) \eqqcolon \N(\rho).
\]
In other words, we trace out the environment's state to obtain the next quantum state of our system. This operation can be viewed as a linear map on density operators, denoted by $\rho \mapsto \N(\rho)$. Such operators are known as super-operators in the literature.

These super-operators are linear and exhibit the following additional properties:
\begin{itemize}
\item[\ding{43}] \textbf{Trace Preservation}: For any linear operator $\rho$ on $\H$ (not necessarily a density operator, so it may not have unit trace), the super-operator $\N$ preserves the trace, meaning $\tr(\N(\rho)) = \tr(\rho)$. This is because super-operators must map density operators to density operators.
\item[\ding{43}] \textbf{Complete Positivity}: For any quantum environment system $\H_e$, the super-operator $\N \otimes \Id$ acting on $\D(\H \otimes \H_e)$ must be positive, meaning it maps positive semi-definite operators to positive semi-definite operators. This property ensures that $\N$ itself is positive, and the combined operator $\N \otimes \Id$ remains positive for any environment quantum system $\H_e$.
\end{itemize}
Conversely, if a super-operator $\N$ is trace-preserving and completely positive, then there exists an environment quantum system $\H_e$ such that $\N$ can be realized in this way (see \cite[Section 8.2]{NiCh02}).

In quantum information theory, super-operators that are trace-preserving and completely positive are called quantum channels. Quantum channels can also connect quantum systems of different dimensions, analogous to classical channels. There are different representations of quantum channels, including the Choi, Kraus, and Stinespring representations. This paper focuses on the Kraus representation (or operator-sum representation): Let \( \N: \D(\H_1) \rightarrow \D(\H_2) \) be a super-operator. It is a quantum channel, meaning it is both trace-preserving and completely positive, if and only if there exist operators \( \{K_l\} \) from \( \H_1 \) to \( \H_2 \) such that (see \cite[Chapter 4]{Wil13})
$$
\N(\rho) = \sum_l K_l \rho K_l^* \text{ and } \sum_l K_l^* K_l = \Id.
$$
Here $\{K_l\}$ are called Kraus operators.

As previously discussed, quantum channels generalize classical channels. Therefore, a key example for a quantum channel is a classical channel. Consider a classical channel \( W: [n] \rightarrow \P([m]) \) that maps each input \( i \in [n] \) to a probability distribution \( W(\cdot|i) \) on \([m]\). This can also be represented as \( W: \P([n]) \rightarrow \P([m]) \), where:
\[
W(\mu)(\cdot) \coloneqq \sum_i W(\cdot|i) \, \mu(i).
\]
Here, we use the notation \( W \) for both mappings.

If we view probability distributions as classical quantum states, the mapping \( W \) corresponds to a super-operator. Specifically, it can be expressed as a quantum channel:
\[
\N_W(\rho) = \sum_{\substack{j \in [m] \\ i \in [n]}} \sqrt{W(j|i)} \, | j \rangle \, \langle i| \, \rho \, |i\rangle \, \langle j | \, \sqrt{W(j|i)}.
\]
Here, \( \{|i\rangle\} \) and \( \{|j\rangle\} \) are orthonormal bases for \( \H_{[n]} \) and \( \H_{[m]} \), respectively. The Kraus operators for \( \N_W \) are \( \{\sqrt{W(j|i)} \, | j \rangle \, \langle i|\}_{i,j} \), satisfying:
\[
\sum_{i,j} (\sqrt{W(j|i)} \, | j \rangle \, \langle i|)^* \, \sqrt{W(j|i)} \, | j \rangle \, \langle i| = \sum_{i,j} |i\rangle \, \langle j | \, \sqrt{W(j|i)} \, \sqrt{W(j|i)} \, | j \rangle \, \langle i| = \Id.
\]
It can be demonstrated that the quantum version of the probability distribution \( W(\mu) \) is the same as \( \N_W(\rho) \), where \( \rho \) is the quantum version of \( \mu \). Therefore, any classical channel can be realized in the quantum domain, confirming that quantum channels generalize classical channels.

Before proceeding to the formulation of q-MDPs, consider a bijective mapping \( f: [n] \rightarrow [n] \), which is a permutation of \([n]\). This mapping defines a classical lossless channel \( W(\cdot|i) \coloneqq \delta_{f(i)}(\cdot) \). The quantum equivalent of this classical channel is given by:
\[
\N_W(\rho) = U \rho U^*,
\]
where \( U \) is the permutation matrix corresponding to the function \( f \). Since \( U \) is a unitary transformation, a classical lossless channel realized by a bijective function corresponds to a closed quantum system. Thus, open quantum systems generalize stochastic evolutions in classical contexts, while closed quantum systems correspond to deterministic (bijective) evolutions in classical settings.

With this foundation, we are now ready to introduce q-MDPs.

\subsection{Quantum Markov Decision Processes}\label{sec2new-2}

Building on the connections between classical and quantum systems discussed in the previous section and using the deterministic reduction of MDPs, we are now ready to introduce the q-MDPs. This quantum model is developed based on the following observations:

\begin{itemize}
\item[\ding{43}] In quantum information and computation, density operators take the place of traditional probability distributions, as mentioned earlier.
\item[\ding{43}] Classical channels, which model probabilistic transitions in classical systems, are generalized by quantum channels in the quantum domain. Quantum channels describe the more intricate interactions and state transformations inherent to quantum systems.
\item[\ding{43}] The expectation operator, used to compute expected values in classical systems, is extended by the trace operator in the quantum framework.
\end{itemize}

With these foundational principles in mind, we can now formally define q-MDPs. This framework enables us to analyze and design decision-making processes in the quantum domain, expanding the capabilities of classical MDPs into quantum settings.

A discrete-time quantum Markov decision process (q-MDP) is characterized by the tuple
\begin{align}
\bigl( \D(\H_{\sX}), \D(\H_{\sX} \otimes \H_{\sA}), \N, C \bigr), \nonumber
\end{align}
where $\D(\H_{\sX})$ denotes the state space, $\D(\H_{\sX} \otimes \H_{\sA})$ is the action space, and $\N$ is the state transition super-operator 
$$
\N: \D(\H_{\sX} \otimes \H_{\sA}) \rightarrow \D(\H_{\sX}),
$$ 
which is a quantum channel. Consequently, given the current action in $\D(\H_{\sX} \otimes \H_{\sA})$, the next state is determined by the quantum channel $\N$, meaning that the state dynamics adhere to the Markov property where only the current action influences the next state, with past states and actions being irrelevant.

The one-stage cost function $C$ is a linear map from $\D(\H_{\sX} \otimes \H_{\sA})$ to $\R$, defined as
$$
C(\sigma) = \sTr(c \sigma) \eqqcolon \langle c,\sigma \rangle
$$
for some Hermitian operator $c$ on $\H_{\sX} \otimes \H_{\sA}$. Here, $\sTr(c \sigma) \eqqcolon \langle c,\sigma \rangle$ represents the Hilbert-Schmidt inner product.

\smallskip

\begin{definition}\label{adms}
A history dependent quantum policy consists of a sequence of quantum channels $\bgamma=\{\gamma_{t}\}$ where
$$
\gamma_t: \left(\bigotimes_{k=0}^t \D\left(\H_{\sX}\right) \right) \otimes \left(\bigotimes_{k=0}^{t-1}\D\left(\H_{\sX} \otimes \H_{\sA}\right)\right) \rightarrow \D(\H_{\sX} \otimes \H_{\sA}),
$$ 
which maps the history $\rho_0 \otimes \sigma_0 \otimes \ldots \otimes \rho_{t-1} \otimes \sigma_{t-1} \otimes \rho_t$ to an action $\sigma_t$ at time $t$. The collection of all such policies is denoted by $\Gamma$. 
\end{definition}

\smallskip

\begin{definition}\label{admsMarkov}
A history dependent quantum policy is Markov if it consists of a sequence of quantum channels $\bgamma=\{\gamma_{t}\}$ from $\H_{\sX}$ to $\H_{\sX} \otimes \H_{\sA}$, meaning that $\gamma_t$ depends solely on the current state $\rho_t$ at each time $t$. The set of all Markov policies is denoted by $\Gamma_m$. A Markov policy $\gamma$ is considered stationary if $\gamma_t = \gamma_{t'}$ for all $t,t'$.
\end{definition}

\smallskip

The state and action evolutions are governed by
\begin{align}
\rho_0 &= \rho_0, \,\,\,\,
\rho_{t+1} = \N(\sigma_t), \text{ } t \geq 0 \nonumber \\
\sigma_t &= \gamma_t(\rho_0 \otimes \sigma_0 \otimes \ldots \otimes \rho_{t-1} \otimes \sigma_{t-1} \otimes \rho_t), \text{ } t \geq 0. \nonumber
\end{align}

This framework can accommodate various cost structures similar to classical MDPs, including finite horizon cost, discounted cost, and average cost criteria, defined as follows:
\begin{align*}
V(\bgamma,\rho_0) &= \sum_{t=0}^{T-1} C(\sigma_t) + C_T(\sigma_T) = \sum_{t=0}^{T-1} \langle c,\sigma_t \rangle + \langle c_T,\sigma_T \rangle \,\, \text{(finite horizon cost)}\\
V(\bgamma,\rho_0) &= \sum_{t=0}^{\infty}\beta^{t} C(\sigma_t) =  \sum_{t=0}^{\infty}\beta^{t} \langle c,\sigma_t \rangle \,\, \text{(discounted cost)}\\
V(\bgamma,\rho_0) &= \limsup_{T \rightarrow \infty} \frac{1}{T}\sum_{t=0}^{T-1} C(\sigma_t) =  \limsup_{T \rightarrow \infty}\frac{1}{T}\sum_{t=0}^{T-1} \langle c,\sigma_t \rangle  \,\, \text{(average cost)}
\end{align*}

While the analysis of q-MDPs can be extended to include finite horizon and average cost criteria, this study primarily focuses on discounted cost. The exploration of finite horizon and average cost criteria is left as a future research direction.

The discounted optimal value of the q-MDP is defined as
\begin{align}
V^*(\rho_0) &\coloneqq \inf_{\bgamma \in \Gamma} V(\bgamma,\rho_0). \nonumber
\end{align}
A policy $\bgamma^{*}$ is considered optimal for $\rho_0$ if $V(\bgamma^{*},\rho_0) = V^*(\rho_0)$.

\smallskip

\begin{remark}
Note that the q-MDP formulation introduced above is based on the idea that it generalizes classical MDPs. As established in our previous work \cite[Section 2.4]{SaSaYu24}, this claim holds true; that is, any classical MDP can be formulated as a q-MDP using its deterministic reduction d-MDP. To avoid repetition, we will not discuss the details here. However, the core idea relies on the connections we established earlier between classical and quantum systems. Specifically, we replace probability distributions in d-MDPs with their diagonal extensions as density operators. Additionally, transition probabilities and policies are replaced by their Kraus representations, as discussed in the previous section. With these substitutions, we derive the quantum formulation of d-MDPs by replacing the expectation operation with the trace operation. For further details, please refer to \cite[Section 2.4]{SaSaYu24}.
\end{remark}

\subsection{Open-Loop and Classical-state-preserving Closed-loop  Policies}

We recall that in d-MDP, a policy is a sequence of classical channels $\bgamma=\{\gamma_{t}\}$ from $\sX$ to $\sX\times\sA$ such that for all $t$
\begin{align} 
\gamma_t(\mu) \in \sA(\mu), \,\, \forall \,\mu \in \P(\sX), \label{inv_classical}
\end{align}
where
$\sA(\mu) \coloneqq \left\{\nu \in \P(\sX \times \sA): v(\,\cdot\,\times\sA) = \mu(\,\cdot\,)\right\}.$ 
It is important to note that we do not impose a condition analogous to (\ref{inv_classical}) on quantum (Markov) policies. Currently, any quantum channel from $\H_{\sX}$ to $\H_{\sX} \otimes \H_{\sA}$ is permissible as a policy. Therefore, the algebraic requirements for a policy are limited to being completely positive and trace-preserving super-operators from $\H_{\sX}$ to $\H_{\sX} \otimes \H_{\sA}$. However, we introduced in \cite{SaSaYu24} variations of the condition (\ref{inv_classical}) specific to the quantum context and define related classes of policies, such as open-loop and classical-state-preserving closed-loop  quantum policies. 

We pursue this for reasons outlined in detail below. In the deterministic reduction of classical MDPs, the key condition that establishes the equivalence between the stochastic MDP model and its deterministic counterpart is the invertibility condition (\ref{inv_classical}). This condition ensures that the state distribution is preserved in the \(\sX\)-marginal of the action of the deterministic reduction. Specifically, if \(\mu_t\) is the state of the deterministic MDP at time \(t\), corresponding to the state distribution of the stochastic MDP at time \(t\), and \(\nu_t\) is the action of the deterministic MDP at time \(t\) under the policy \(\gamma_t\), the condition requires that:
\[
\nu_t(\cdot \times \sA) = \mu_t(\cdot).
\]
This preservation of state information in the action marginal is crucial.

In contrast, Definitions~\ref{adms} and \ref{admsMarkov} place no explicit restrictions on the policies to ensure that state information is retained in the action. Under both history dependent quantum policies and Markov quantum policies, any current quantum state \(\rho\) can be mapped to any quantum action \(\sigma\) via some specific quantum channel. Indeed, this can be achieved using a combination of the \emph{trace-out channel} and the \emph{preparation channel} \cite[p. 143]{Wil13}. Consequently, this formulation is highly flexible, allowing the system to achieve virtually any desired outcome, as the one-stage cost depends only on the current state information \(\sigma_t\). Moreover, the current policy formulation fails to preserve state information within the action, which directly conflicts with a foundational principle of control theory: state feedback. In control theory, state feedback is essential for designing effective control strategies, as it relies on using the system's current state to determine the appropriate action. By omitting state information from the action, the policy loses its ability to adapt dynamically to the system's evolving conditions. 

For the reasons outlined above, our goal is to establish a condition analogous to (\ref{inv_classical}) in the quantum setting, mirroring the formulation used in the classical case. This approach allows the quantum framework to be seen as an extension of its classical counterpart. In d-MDPs, the invertibility condition on policies ensures the equivalence between classical MDPs and their deterministic reduction. Similarly, in the quantum domain, introducing analogous constraints preserves a direct correspondence between classical and quantum Markov decision processes. This is essential for developing a unified framework that seamlessly incorporates both classical and quantum domains. 

Additionally, by aligning the quantum formulation with the classical structure, we not only ensure continuity in theoretical development but also make it possible to adapt well-established classical techniques for the quantum context. Moreover, incorporating conditions like (\ref{inv_classical}) enables the realization of the most fundamental concept of control theory --- state feedback. Specifically, as we will demonstrate, an appropriate invertibility condition on the policy in the quantum domain allows the action variable to retain partial or complete information about the state. This facilitates the controller's ability to utilize feedback from the state variable, thereby enhancing the control framework.

In d-MDPs, the condition (\ref{inv_classical}) can be rephrased as follows: the marginalization channel $\Mar(\nu) \coloneqq \nu(\cdot \times \sA),$ where $\Mar: \P(\sX \times \sA) \rightarrow \P(\sX)$, serves as the inverse of the channel $\gamma_t$; that is, $\Mar \circ \gamma_t(\mu) = \mu \text{ for all } \mu \in \P(\sX).$ This relationship can be written as $\Mar = \Inv(\gamma_t)$, where the inverse operator $\Inv$ acts on classical channels to reverse their effects.

In the next two sub-sections, we investigate various frameworks for the operator $\Inv$ and the quantum channel $\Mar$, progressively generalizing above classical operators to the quantum domain. We begin by examining the most straightforward extensions of the classical operations $\Inv$ and $\Mar$ to their quantum analogs, which naturally leads to the definition of open-loop quantum policies. Following this, we relax the definition of the operator $\Inv$ to derive the classical-state-preserving closed-loop  policies.

\paragraph{Open-loop Quantum Policies}

Here, we introduce the initial descriptions of the $\Inv$ and $\Mar$ operations, which extend their classical counterparts to the quantum domain. First, the inverse operator $ \Inv:\N_{(\H_{\sX} \rightarrow \H_{\sX} \otimes \H_{\sA})} \rightarrow \N_{(\H_{\sX} \otimes \H_{\sA} \rightarrow \H_{\sX})}, $ which naturally generalizes the classical inverse operator, is defined as the canonical inverse of any quantum channel. Specifically, for any $\gamma \in \N_{(\H_{\sX} \rightarrow \H_{\sX} \otimes \H_{\sA})}$, we define $\Inv(\gamma) \eqqcolon \gamma^{-1} \in \N_{(\H_{\sX} \otimes \H_{\sA} \rightarrow \H_{\sX})}$ as the quantum channel that satisfies: $ \gamma^{-1} \circ \gamma(\rho) = \rho \text{ for all } \rho \in \D(\H_{\sX}). $ To distinguish this from other inverse introduced later, we refer to it as $\Inv_{\C}$.

The marginalization channel is defined based on the observation that the quantum analogue of the classical marginalization operation is the partial trace operator. Therefore, we define $\Mar$ as the partial trace channel $\tr_{\sA}: \D(\H_{\sX} \otimes \H_{\sA}) \rightarrow \D(\H_{\sX})$.

\smallskip

\begin{definition}
With these specific descriptions for $(\Inv,\Mar)$, an open-loop quantum policy is defined to be a sequence quantum channels $\bgamma=\{\gamma_{t}\}$ from $\H_{\sX}$ to $\H_{\sX} \otimes \H_{\sA}$ such that, for all $t$, 
$
\tr_{\sA} = \Inv_{\C}(\gamma_t), 
$
or equivalently
$
\gamma_t(\rho) \in \sA(\rho), \, \forall \rho \in \D(\H_{\sX}).
$
Let $\Gamma_o$ denote the set of open-loop policies.
\end{definition}

The following structural result is established in \cite[Proposition 5.4]{SaSaYu24}, which clarifies why we refer to such policies as \emph{open-loop}. 

\begin{proposition}\label{structure-qmdp}
Let $\gamma: \D(\H_{\sX}) \rightarrow \D(\H_{\sX}\otimes\H_{\sA})$ be a quantum channel. It satisfies the reversibility condition: $\tr_{\sA}(\gamma(\rho)) = \rho$ for all $\rho \in \D(\H_{\sX})$ if and only if there exists $\xi \in \D(\H_{\sA})$ such that $\gamma(\rho) = \rho  \otimes \xi$ for all $\rho \in \D(\H_{\sX})$.
\end{proposition}

In the literature, the channels described in Proposition~\ref{structure-qmdp} are often called \emph{appending channels} \cite{Wil13}. According to this proposition, for any policy $\bgamma = {\gamma_t}$ that meets the invertibility condition, there exists a sequence of density operators ${\pi_t} \subset \D(\H_{\sA})$ such that $\gamma_t(\rho) = \rho \otimes \pi_t$ for every $t$. Thus, one can alternatively represent the policy $\bgamma$ using this sequence $\{\pi_t\}$.

Given this structural result, these quantum policies do not modify the state information through any quantum operation to extract additional information from the state. Instead, the policy relies entirely on the information provided by the initial density operator. As a result, these policies are classified as \emph{open-loop quantum policies}, analogous to open-loop policies in classical control theory. However, this does not imply that the policies completely disregard feedback from the state. Indeed, as we will see in the dynamic programming principle developed for open-loop policies, the appended optimal density operator depends on the current state. This means that while the state feedback is used, the state is not modified via quantum operations to extract some information from it.

Before introducing classical-state-preserving closed-loop policies, we first discuss how classical policies can be embedded within the quantum domain. These quantum analogs of classical policies are referred to as classical-quantum policies.

\paragraph{Classical Quantum Policies}

In Section~2.4 of our previous work \cite[Section 2.4]{SaSaYu24}, we provided a detailed explanation of how a d-MDP can be transformed into a q-MDP. As part of this transformation, we defined a classical policy within the quantum domain as follows.

\begin{definition}
An admissible classical policy is a sequence quantum channels $\bgamma=\{\gamma_{t}\}$ from $\H_{\sX}$ to $\H_{\sX} \otimes \H_{\sA}$, where, for all $t$, there exists a classical channel $\gamma_t: \P(\sX) \rightarrow \P(\sX\times\sA)$ realized by the stochastic kernel $\pi_t$ from $\sX$ to $\sA$ such that 
$$
\gamma_t(\rho) = \sum_{(x,a) \in \sX \times \sA} \sqrt{\pi_t(a|x)} \, |x,a\rangle \, \langle x| \, \rho |x\rangle \, \langle x,a| \, \sqrt{\pi_t(a|x)}. 
$$
The set of classical policies is denoted by $\Gamma_c$.
\end{definition}

For classical quantum policies, it is evident that 
\[
\tr_{\sA} \circ \gamma_t(\rho) = \rho
\]
does not hold for all \(\rho \in \D(\H_{\sX})\), implying that \(\tr_{\sA} \neq \Inv_{\C}(\gamma_t)\). As a result, classical quantum policies do not form a subset of open-loop quantum policies. Consequently, the current formulation of q-MDPs with open-loop policies does not directly extend classical MDPs. 

To address this limitation, we propose relaxing the reversibility constraint by introducing a different inverse operator \(\Inv\) in place of \(\Inv_{\C}\). This modification allows the quantum counterparts of classical policies to satisfy the relaxed invertibility condition, thereby enabling q-MDPs to serve as a true generalization of d-MDPs. This is achieved by introducing classical-state-preserving closed-loop policies, as discussed in the next section.

\paragraph{Classical-state-preserving Closed-loop  Quantum Policies}

Now, we define classical-state-preserving closed-loop  quantum policies for q-MDPs, which allow the controller to use current state information by modifying it through a quantum operation as opposed to the open-loop policies. This is achieved by relaxing the reversibility condition inherent in open-loop policies. Essentially, we propose a different version of the operator $\Inv$. 

The alternative inverse operator is defined as follows:

\begin{definition}\label{inverse1}
The inverse operator  
\[
\Inv_{\S}:\N_{(\H_{\sX} \rightarrow \H_{\sX} \otimes \H_{\sA})} \rightarrow \N_{(\H_{\sX} \otimes \H_{\sA} \rightarrow \H_{\sX})}
\]
is defined such that for any $\gamma \in \N_{(\H_{\sX} \rightarrow \H_{\sX} \otimes \H_{\sA})}$, $\Inv_{\S}(\gamma) \eqqcolon \gamma^{-1} \in \N_{(\H_{\sX} \otimes \H_{\sA} \rightarrow \H_{\sX})}$ satisfies 
\[
\gamma^{-1} \circ \gamma(\rho) = \rho, \, \forall \rho \in \S,
\]
where $\S = \{|x\rangle \langle x| : x \in \sX\}$ and $\{|x\rangle\}$ is some fixed orthonormal basis of $\H_{\sX}$. 
\end{definition}

The orthonormal basis described above should be understood as the set of classical pure states of a quantum system \cite[Section 4.1.4]{Wil13}.

Since the convex closure of $\S$ is a strict subset of $\D(\H_{\sX})$, this requirement is less stringent than the previous reversibility condition. Introducing this new inverse operator defines a new set of admissible policies called classical-state-preserving closed-loop  policies, a term that will be clarified as we state their structural properties.

\begin{definition}\label{closed1}
A \emph{classical-state-preserving closed-loop  policy} is a sequence of quantum channels $\bgamma=\{\gamma_{t}\}$ from $\H_{\sX}$ to $\H_{\sX} \otimes \H_{\sA}$ such that, for all $t$, $\tr_{\sA}= \Inv_{\S}(\gamma_t)$; that is, 
\[
\tr_{\sA}(\gamma_t(\rho)) = \rho, \,\,\, \forall \rho \in \S.
\]
\end{definition}

Let $\cC_w$ denote the set of quantum channels from $\H_{\sX}$ to $\H_{\sX}\otimes\H_{\sA}$ that satisfy the relaxed reversibility constraint. We can define the set of classical-state-preserving closed-loop  quantum policies as 
\[
\Gamma_w \coloneqq \{\bgamma=\{\gamma_{t}\}: \gamma_t \in {\cC}_w \,\, \forall t\}.
\]
Thus, $\Gamma_o \subset \Gamma_w$. Moreover, quantum adaptations of classical policies satisfy this relaxed condition. 
 Hence, $\Gamma_c \subset \Gamma_w$ (see \cite[Section 5.2]{SaSaYu24} for details). Consequently, q-MDPs with classical-state-preserving closed-loop  policies truly generalize d-MDPs. We refer to these q-MDPs as qw-MDPs to distinguish them from q-MDPs with open-loop policies.

The following proposition, whose proof can be found in \cite[Proposition 5.8]{SaSaYu24}, describes the structure of classical-state-preserving closed-loop  quantum policies:

\begin{proposition}\label{structure-qwmdp}
Let $\gamma: \D(\H_{\sX}) \rightarrow \D(\H_{\sX}\otimes\H_{\sA})$ be a quantum channel. It belongs to $\cC_w$ if and only if there exists a collection of vectors $\{|\phi_{x,a} \rangle\}_{(x,a)\in\sX\times\sA}$ in some Hilbert space $\H_L$ with $\dim(\H_L) \leq |\sX|^2|\sA|$ such that $\sum_{a \in \sA} \langle \phi_{x,a},\phi_{x,a} \rangle=1$ for each $x \in \sX$ and 
\begin{align*}
\gamma(|\psi\rangle \langle \psi|) = \sum_{x,y \in \sX} \langle x,\psi \rangle \langle \psi,y \rangle |x\rangle \langle y| \otimes \left(\sum_{a,b \in \sA} \langle \phi_{y,b},\phi_{x,a} \rangle \, |a\rangle \langle b | \right)
\end{align*}
for each pure state $|\psi\rangle \in \H_{\sX}$.
\end{proposition}

Policies of the form 
\begin{align*}
\gamma(|\psi\rangle \langle \psi|) = \sum_{x,y \in \sX} \langle x,\psi \rangle \langle \psi,y \rangle |x\rangle \langle y| \otimes \left(\sum_{a,b \in \sA} \langle \phi_{y,b},\phi_{x,a} \rangle \, |a\rangle \langle b | \right)
\end{align*}
clearly extract some information from the state via the inner products $\langle \phi_{y,b},\phi_{x,a} \rangle$. Thus, these policies are classified as \emph{classical-state-preserving closed-loop  quantum policies}.

In the rest of this paper, we focus on developing methodologies to achieve optimal policies and value functions for both open-loop and classical-state-preserving closed-loop  quantum policies. In classical MDPs, dynamic programming and linear programming offer distinct but closely related methods for determining optimal policies and value functions, with each method being the dual of the other \cite{HeLa96}. We extend these classical approaches to the quantum domain, leveraging their duality to create methodologies for deriving the optimal open-loop and classical-state-preserving closed-loop  quantum policies.

\begin{remark}\label{realizability} {\bf{[On Physical Realizability of the Classes of Policies].}} 
Although we have introduced several policies for q-MDPs, we do not address whether all of these policies are realizable in a physical (quantum mechanical) sense. Our primary goal in this paper here is to propose a general mathematical framework for quantum MDPs; however, we note that history-dependent and Markov policies can be interpreted as physically non-realizable due to the fact that they do not preserve state information in the action variable: because we do not impose any invertibility condition, given any state $\rho$, we can theoretically steer the state to any action $\sigma$ by first applying a trace-out channel followed by a preparation channel \cite[p. 143]{Wil13}. While this is a useful theoretical construct, practical constraints make it generally non-realizable, and further cost terms may be incorporated into the formulation. For instance, state preparation is a well-known challenge in quantum information processing \cite{Gri09}, and freely applying a preparation channel is often infeasible.
Moreover, in this policy formulation, since the one stage cost function depends only on the action $\sigma$ and we can steer the state to any $\sigma$ at each time step, the optimal control problem becomes trivial by identifying an action $\sigma^*$ that minimizes the one-stage cost and steering the action to $\sigma^*$ at every step. To avoid trivial results, we must introduce additional structure. Specifically, we need to incorporate a penalty on the state variable, as the action variable alone does not preserve any state information due to the lack of invertibility in the history dependent and Markov policies. This can be done by defining one-stage cost functions that depend on both the current state $\rho$ and the current action $\sigma$ or by adding a penalty term to perturbations on the state (as measured by the difference between the state and the trace-out on the action). Such a formulation prevents the system from arbitrarily steering to any action and leads to interesting further research problems. Nonetheless, the essence of the mathematical formulation would not change, and therefore, our general formulation would still be applicable.

\end{remark}

\section{Algorithms for Open-loop Quantum Policies}\label{sec2}

Recall that an open-loop quantum policy is defined to be a sequence quantum channels $\bgamma=\{\gamma_{t}\}$ from $\H_{\sX}$ to $\H_{\sX} \otimes \H_{\sA}$ such that, for all $t$, 
$
\tr_{\sA} = \Inv_{\C}(\gamma_t), 
$
or equivalently
$
\gamma_t(\rho) \in \sA(\rho), \, \forall \rho \in \D(\H_{\sX}).
$
It is established that for any open-loop policy $\bgamma = \{\gamma_t\}$, there exists a sequence of density operators $\{\pi_t\} \subset \D(\H_{\sA})$ such that $\gamma_t(\rho) = \rho \otimes \pi_t$ for all $t$. Hence, in view of this, one can alternatively identify the policy $\bgamma$ via this sequence $\{\pi_t\}$. 

In the rest of this section, we follow a similar approach to the one used for d-MDP in the appendix. Namely, we initially derive the dynamic programming principle for open-loop policies and then formulate q-MDP with open-loop policies as a semi-definite program (SDP). Subsequently, we establish a connection between the dual SDP and the dynamic programming equation, which leads to a linear optimal value function. Leveraging this linearity, we demonstrate the existence of an optimal stationary open-loop policy. 

We note that in order to establish a connection with the methodology employed for d-MDP in the appendix, we utilize the same notation to represent operators that yield the dynamic programming equation and semi-definite program (SDP) in q-MDP. This approach allows us to draw parallels between the two frameworks and facilitate a seamless connection in our analysis.

\subsection{Dynamic Programming for q-MDP with Open-loop Policies}\label{sub2sec2}

The dynamic programming operator for q-MDP with open-loop policies, denoted as 
$\L:C_b(\D(\H_{\sX})) \rightarrow C_b(\D(\H_{\sX}))$
is used for generating both the optimal value function and the optimal open-loop policy. It is defined as follows:
$
\L V(\rho) \coloneqq \min_{\pi \in \D(\H_{\sA})} \left[ \langle c,\rho \otimes \pi \rangle + \beta \, V(\N(\rho \otimes \pi )) \right].
$
First of all, it is straightforward to prove that $\L$ maps continuous function to continuous function \cite[Proposition D.6]{HeLa96} as $\D(\H_{\sA})$ is compact (because $\H_{\sA}$ is finite dimensional) and the function $\langle c,\rho \otimes \pi \rangle + \beta \, V(\N(\rho \otimes \pi ))$ is continuous in $\rho$ and $\pi$ if $V$ is continuous. 
Moreover, one can also show that the operator $\L$ is $\beta$-contraction with respect sup-norm on $C_b(\D(\H_{\sX}))$, and so, has a unique fixed point $V_{\mathrm{f}} \in C_b(\D(\H_{\sX}))$. In the following, we establish a connection between the optimal value function $V^*$ of q-MDP with open-loop policies and $V_{\mathrm{f}}$. 

\begin{theorem}\label{dynamic-qmdp}
We have $V^*(\rho_0) = V_{\mathrm{f}}(\rho_0)$ for all $\rho_0 \in \D(\H_{\sX})$; that is, the unique fixed point of $\L$ is $V^*$.
\end{theorem}

\begin{proof}
Fix any open-loop policy $\bgamma = \{\gamma_t\}_{t\geq0}$ with the corresponding density operators $\{\pi_t\}_{t\geq0}$; that is, $\gamma_t(\rho) = \rho \otimes \pi_t$. Then, the cost function of $\bgamma$ satisfies the following 
\begin{align*}
V(\bgamma,\rho_0) &= \langle c,\rho_0 \otimes \pi_0 \rangle + \beta \, V(\{\gamma_t\}_{t\geq 1},\N(\rho_0\otimes\pi_0)) \\
&\geq \min_{\pi_0 \in \D(\H_{\sA})} \left[ \langle c,\rho_0 \otimes \pi_0 \rangle + \beta \, V(\{\gamma_t\}_{t\geq 1},\N(\rho_0\otimes\pi_0)) \right] \\
&\geq \min_{\pi_0 \in \D(\H_{\sA})} \left[ \langle c,\rho_0 \otimes \pi_0 \rangle + \beta \, V^*(\N(\rho_0\otimes\pi_0)) \right] = \L V^*(\rho_0).
\end{align*}
Since above inequality is true for any open-loop policy $\bgamma$, we have $V^* \geq \L V^*$. Using the latter inequality and the monotonicity of $\L$, one can also prove that $V^* \geq \L^nV^*$ for any positive integer $n$. By Banach fixed point theorem, as the operator $\L$ is $\beta$-contraction with respect sup-norm on $C_b(\D(\H_{\sX}))$, we have $\L^n V^* \rightarrow V_{\mathrm{f}}$ in sup-norm in $C_b(\D(\H_{\sX}))$ as $n\rightarrow \infty$. Hence, $V^* \geq  V_{\mathrm{f}}$. 

To prove the converse (that is, achievability), for any $\rho$, let 
$$
\pi_{\rho} \in \argmin_{\pi \in \D(\H_{\sA})} \left[ \langle c,\rho \otimes \pi \rangle + \beta \, V_{\mathrm{f}}(\N(\rho \otimes \pi )) \right]
$$
The existence of $\pi_{\rho}$ for any $\rho$ follows from the facts that $\D(\H_{\sA})$ is compact and  $\langle c,\rho \otimes \pi \rangle + \beta \, V_{\mathrm{f}}(\N(\rho\otimes\pi))$ is continuous in $\rho$ and $\pi$. With this construction, we then have
\begin{align*}
V_{\mathrm{f}}(\rho_0) &=  \langle c,\rho_0 \otimes \pi_{\rho_0} \rangle + \beta \, V_{\mathrm{f}}(\N(\rho_0 \otimes \pi_{\rho_0} )) \\
&\eqqcolon \langle c,\sigma_0  \rangle + \beta \, V_{\mathrm{f}}(\rho_1) \\
&= \langle c,\sigma_0  \rangle + \beta \, \left[\langle c,\rho_1 \otimes \pi_{\rho_1} \rangle + \beta \, V_{\mathrm{f}}(\N(\rho_1 \otimes \pi_{\rho_1} )) \right] \\
&\eqqcolon \langle c,\sigma_0  \rangle + \beta \, \left[\langle c,\sigma_1 \rangle + \beta \, V_{\mathrm{f}}(\rho_2) \right] \\
&\phantom{x}\vdots \\
&= \sum_{t=0}^{N-1} \beta^t \, \langle c,\sigma_t \rangle + \beta^N \, V_{\mathrm{f}}(\rho_N) \rightarrow V(\{\gamma_t\},\rho_0) \,\, \text{as} \,\ N\rightarrow\infty
\end{align*}
where $\gamma_t(\rho) \coloneqq \rho \otimes \pi_{\rho_t}$ for each $t$. Hence, $V_{\mathrm{f}} \geq V^*$. This completes the proof.
\end{proof}

If one examines of the second part of the proof, it can be seen that the optimal open-loop policy $\bgamma^*$ for any given initial density operator $\rho_0$ can be constructed as follows: $\gamma^*_0(\rho) = \rho \otimes \pi_{\rho_0}$ and $\gamma^*_{t+1}(\rho) = \rho \otimes \pi_{\rho_{t+1}}$ for $t\geq0$, where $\rho_{t+1} = \N(\rho_t \otimes \pi_{\rho_t})$. Therefore, the following mapping 
$$f: \D(\H_{\sX}) \ni \rho \mapsto \pi_{\rho} \in \D(\H_{\sA})$$ 
gives the optimal actions for any state in the open-loop case. However, this mapping $f$ cannot be viewed as a stationary open-loop policy, as it does not yield a valid appending quantum channel. This is because it requires adjusting the appending action based on the input state each time, which disqualifies it from being an appending quantum channel. Recall that appending channel must be of the following form 
$$
\D(\H_{\sX}) \ni \rho \mapsto \rho \otimes \pi \in \D(\H_{\sX} \otimes \H_{\sA})
$$ 
for some fixed $\pi \in \D(\H_{\sA})$. Indeed, in Theorem~\ref{quantum-stationary}, we establish that there exists an optimal stationary open-loop policy. Namely, there exists a fixed density operator $\pi^* \in \D(\H_{\sA})$ such that the following appending channel 
$$
\D(\H_{\sX}) \ni \rho \mapsto \rho \otimes \pi^* \in \D(\H_{\sX} \otimes \H_{\sA})
$$ 
is optimal; that is, regardless of the input state $\rho$, the output action is $\rho \otimes \pi^*$ instead of $\rho \otimes \pi_{\rho}$ as discussed above. This result is achieved through the duality between the SDP formulation and the dynamic programming formulation of the optimization problem. This marks a primary distinction between our approach and previous formulations in the literature \cite{BaBaAa14, YiYi18, YiFeYi21}, where we require admissible policies to be quantum channels and exclude classical elements from our framework.

The discussion above on the optimality of the open-loop quantum policies implies that to determine the optimal open-loop policy for $\rho_0$, knowledge of the optimal value function $V^*$ is required. Furthermore, it is observed that the optimal open-loop policy depends on the initial density operator and for now changes over time, i.e., it is non-stationary. However, as discussed above, we will remove this non-stationary constraint later.

Indeed, in the subsequent section, a semi-definite programming formulation of q-MDP reveals that the optimal value function $V^*(\rho)$ can be expressed as a linear function of $\rho$, denoted by $\langle \xi^*,\rho \rangle$, where $\xi^*$ is a Hermitian operator on $\H_{\sX}$.  Using this, it will be proven that an optimal stationary open-loop policy exists.

Note that the latter part of the above proof also implies that the following condition is sufficient for optimality of any open-loop policy: 

\begin{tcolorbox} 
[colback=white!100]
\begin{itemize}
\item[ ] {\bf (Opt):} If the open-loop policy $\bgamma = \{\gamma_t\}$ with the corresponding density operators $\{\pi_t\} \subset \D(\H_{\sA})$ satisfies the following conditions
\begin{align*}
&\bullet \, \rho_{t+1} = \N(\rho_t \otimes \pi_t), \,\, \forall t\geq0 \\
&\bullet \, \pi_t \in \argmin_{\pi \in \D(\H_{\sA})} \left[ \langle c,\rho_t \otimes \pi \rangle + \beta \, V^*(\N(\rho_t \otimes \pi )) \right]
\end{align*}
then $\bgamma$ is optimal.
\end{itemize}
\end{tcolorbox}

\subsection{SDP Formulation of q-MDP with Open-loop Policies}\label{sub3sec2}

Given any open-loop policy $\bgamma$, define the state-action occupation density operator as follows

\[\sigma^{\bgamma} \coloneqq (1-\beta) \, \sum_{t=0}^{\infty} \beta^t \, \sigma_t,\]

where the term $(1-\beta)$ is the normalization constant. Here, it is straightforward to prove that the series $\sum_{t=0}^{\infty} \beta^t \, \sigma_t$ converges in Hilbert-Schmidt norm \cite[Definition 4.4.5]{Wil13}; that is, the partial sum sequence $\{\sum_{t=0}^{n} \beta^t \, \sigma_t\}_{n\geq0}$ is convergent. Note that one can write the normalized cost of $\bgamma$ in the following form
$
(1-\beta) \, V(\bgamma,\rho_0) =  C(\sigma^{\bgamma}) = \langle c,\sigma^{\bgamma} \rangle
$
as a result of linearity of $C$. Let $\rho^{\bgamma} \coloneqq \tr_{\sA}(\sigma^{\bgamma})$. We then have 
\begin{align}
\rho^{\bgamma}(\,\cdot\,) &= (1-\beta) \, \sum_{t=0}^{\infty} \beta^t \, \tr_{\sA}(\sigma_t) \,\, (\text{by linearity of $\tr_{\sA}$})\nonumber \\
&= (1-\beta) \, \sum_{t=0}^{\infty} \beta^t \, \tr_{\sA}(\gamma_t(\rho_t)) \nonumber \\
&= (1-\beta) \, \sum_{t=0}^{\infty} \beta^t \rho_t \,\, \text{(as $\Inv_{\C}(\gamma_t) = \tr_{\sA}$)} \nonumber \\
&= (1-\beta) \,\rho_0 + (1-\beta) \, \beta \sum_{t=1}^{\infty} \beta^{t-1} \, \N(\sigma_{t-1}) \nonumber \\
&= (1-\beta) \, \rho_0 + \beta \, \N\left((1-\beta) \,\sum_{t=1}^{\infty} \beta^{t-1} \, \sigma_{t-1}\right)  \,\, \text{(by linearity of $\N$)}\nonumber \\
&= (1-\beta) \,\rho_0 + \beta \, \N(\sigma^{\bgamma}).\nonumber 
\end{align}
This computation suggests the definition of the Hermicity preserving super-operator $\T: \cL_H(\H_{\sX}\otimes\H_{\sA}) \rightarrow \cL_H(\H_{\sX})$ as follows:
$
\T(\sigma) \coloneqq \tr_{\sA}(\sigma) - \beta \, \N(\sigma),
$
$\cL_H(\H)$ denotes the set of all Hermitian operators on the Hilbert space $\H$.
Therefore, $\T(\sigma^{\bgamma}) = (1-\beta) \,\rho_0$ for any open-loop policy $\bgamma$. This observation leads to the following semi-definite programming formulation:
\begin{tcolorbox} 
[colback=white!100] 
\begin{align}
(\textbf{SDP}) \text{                         }&\min_{\sigma \in \cL_H(\H_{\sX}\otimes\H_{\sA})} \text{ } \langle c,\sigma \rangle
\nonumber \\*
&\text{s.t.} \, \T(\sigma) = (1-\beta) \,\rho_0 \,\, \text{and} \,\, \sigma \succcurlyeq 0. \nonumber 
\end{align} 
\end{tcolorbox}
Here $\sigma \succcurlyeq 0$ means that $\sigma$ is positive semi-definite. The partial order $\succcurlyeq$ between square operators is called L\"{o}wner order in the literature \cite{KhWi20}. Note that if $\T(\sigma) = (1-\beta) \,\rho_0$ and $\sigma \succcurlyeq 0$, then necessarily $\sigma$ is a density operator. As a result, it is unnecessary to impose a prior restriction on $\sigma$ to be confined within $\D(\H_{\sX}\otimes\H_{\sA})$ in the semi-definite programming formulation. 

Unfortunately, unlike in the classical case (see Section~\ref{sub3sec1}), it is not possible to directly prove the equivalence between $(\textbf{SDP})$ and q-MDP with open-loop policies. This is because, while it is known that for any open-loop policy $\bgamma$, there exists a feasible $\sigma^{\bgamma} \in \cL_H(\H_{\sX}\otimes\H_{\sA})$ for $(\textbf{SDP})$ such that $(1-\beta) \, V(\bgamma,\rho_0) = \langle c,\sigma^{\bgamma} \rangle$, it cannot be proven that if $\sigma \in \cL_H(\H_{\sX}\otimes\H_{\sA})$ is feasible for $(\textbf{SDP})$, then there exists a open-loop policy $\bgamma^{\sigma}$ such that $(1-\beta) \, V(\bgamma^{\sigma},\rho_0) = \langle c,\sigma \rangle$. The reason for this is that it is not clear a priori whether any feasible $\sigma$ can be represented as a tensor product, which is necessary to extract a control policy as an appending channel (see the converse part of the proof of Theorem~\ref{lpformulation}). As a result, we can only conclude that $(\textbf{SDP})$ serves as a lower bound for q-MDP based on this argument. However, by utilizing the dual of $(\textbf{SDP})$ and incorporating additional condition, we can establish their equivalence. Furthermore, we can prove that the optimal value function $V^*$ is linear and that a stationary optimal open-loop policy exists.
 
Above discussion establishes that 
\begin{align*}
\inf_{\bgamma \in \Gamma} (1-\beta) \, V(\bgamma,\rho_0) \geq (\textbf{SDP}) \text{                         }&\min_{\sigma \in \cL_H(\H_{\sX}\otimes\H_{\sA})} \text{ } \langle c,\sigma \rangle
\nonumber \\*
&\text{s.t.} \, \T(\sigma) = (1-\beta) \,\rho_0 \,\, \text{and} \,\, \sigma \succcurlyeq 0.
\end{align*} 
In the subsequent part of this section, it is demonstrated that the optimal value function $V^*$ is linear in $\rho$, and based on this, it is proven that there exists an optimal stationary open-loop policy.

First of all, let us obtain the dual of the above SDP. Indeed, we have
\begin{align*} 
(\textbf{SDP})
& = \min_{\sigma \in \cL_H(\H_{\sX}\otimes\H_{\sA})} \text{ } \langle c,\sigma \rangle
\,\, \text{s.t.} \, \T(\sigma) = (1-\beta) \,\rho_0 \,\, \text{and} \,\, \sigma \succcurlyeq 0 \nonumber \\
&=\min_{\sigma \succcurlyeq 0} \, \max_{\xi \in \cL_H(\H_{\sX})} \, \langle c,\sigma \rangle  + \langle \xi, (1-\beta) \,\rho_0 -\T(\sigma) \rangle \\
&\text{(by Sion min-max theorem)} \\
&= \max_{\xi \in \cL_H(\H_{\sX})} \, \min_{\sigma \succcurlyeq 0} \, \langle c,\sigma \rangle  + \langle \xi, (1-\beta) \,\rho_0 -\T(\sigma) \rangle \\
&= \max_{\xi \in \cL_H(\H_{\sX})} \, \min_{\sigma \succcurlyeq 0} \, \langle \xi,(1-\beta) \, \rho_0 \rangle  + \langle c-{\T}^{\dag}(\xi), \sigma \rangle \\
&=\max_{\xi \in \cL_H(\H_{\sX})} \, \langle \xi,(1-\beta) \, \rho_0 \rangle 
\,\, \text{s.t.} \, c-{\T}^{\dag}(\xi)  \succcurlyeq  0 
\eqqcolon (\textbf{SDP}^{\dag}), 
\end{align*} 
where ${\T}^{\dag}:\cL_H(\H_{\sX}) \rightarrow \cL_H(\H_{\sX}\otimes\H_{\sA})$ is the adjoint of $\T$ and is given by
$
{\T}^{\dag}(\xi) = \xi \otimes \Id - \beta \, \N^{\dag}(\xi).
$
For the dual SDP, we have the following equivalent formulations:\footnote{In Section~\ref{sub3sec1}, we discuss several equivalent formulations of the dual linear program for d-MDPs. However, when it comes to alternative formulations for q-MDPs, the dual SDP falls short in richness. This is why we require an additional assumption and a distinct analysis to prove the linearity of $V^*$ and the stationarity of the optimal open-loop policy in q-MDPs.}
\begin{tcolorbox} 
[colback=white!100]
\begin{align*}
&(\textbf{SDP}^{\dag}) \\
&=\max_{\xi \in \cL_H(\H_{\sX})} \text{ }  \langle \xi,(1-\beta) \, \rho_0 \rangle 
\,\, \text{s.t.}  \, c + \beta \, \N^{\dag}(\xi) \succcurlyeq \xi \otimes \Id \\
&=\max_{\xi \in \cL_H(\H_{\sX})} \text{ }  \langle \xi,(1-\beta) \, \rho_0 \rangle 
\,\, \text{s.t.}  \, \langle c + \beta \, \N^{\dag}(\xi), |\psi\rangle \langle \psi | \rangle \geq \langle \xi\otimes \Id,|\psi\rangle \langle \psi | \rangle \\
&\phantom{xxxxxxxxxxxxxxxxxxxxxxxxxxxxxxxxx} \text{for all pure states} \, |\psi \rangle \in \H_{\sX}\otimes\H_{\sA} \\
&=\max_{\xi \in \cL_H(\H_{\sX})} \text{ }  \langle \xi,(1-\beta) \, \rho_0 \rangle 
\,\, \text{s.t.}  \, \langle c + \beta \, \N^{\dag}(\xi), \sigma \rangle \geq \langle \xi\otimes \Id,\sigma \rangle \,\, \forall \, \sigma \in \D(\H_{\sX}\otimes\H_{\sA}).
\end{align*}  
\end{tcolorbox}
Now, we are about to present our supplementary assumption. In the appendix, we establish a criterion that is both necessary and sufficient to validate this assumption.

\smallskip

\begin{assumption}\label{as1}
There exists a feasible $\xi^* \in \cL_H(\H_{\sX})$ for $(\textbf{SDP}^{\dag})$ such that for any $\rho \in \D(\H_{\sX})$, we have 
$
\min_{\pi \in \D(\H_{\sA})} \langle c + \beta \, \N^{\dag}(\xi^*), \rho \otimes \pi \rangle = \langle \xi^*,\rho \rangle.
$
\end{assumption}

\smallskip

\begin{remark}
Assumption~\ref{as1} is not overly restrictive. Consider any $\rho \in \D(\H_{\sX})$. If $\xi$ is a feasible solution for $(\textbf{SDP}^{\dag})$, it automatically satisfies
\[
\langle c + \beta \, \N^{\dag}(\xi), \rho \otimes \pi \rangle \geq \langle \xi\otimes \Id,\rho \otimes \pi \rangle
\]
for any $\pi \in \D(\H_{\sA})$. Note that $\langle \xi\otimes \Id,\rho \otimes \pi \rangle = \langle \xi,\rho \rangle$ and so above feasibility condition can be written as:
\[
\langle c + \beta \, \N^{\dag}(\xi), \rho \otimes \pi \rangle \geq \langle \xi,\rho \rangle.
\]
In Assumption~\ref{as1}, we require the existence of a feasible solution $\xi^*$ for $(\textbf{SDP}^{\dag})$ that achieves equality in this inequality for any $\rho$, when the left side is minimized over $\pi \in \D(\H_{\sA})$. 
\end{remark}

Let us define a linear function $\hat{V}(\rho) \coloneqq \langle \xi^*,\rho \rangle$ on density operators $\rho \in \D(\H_{\sX})$. Then for any $\rho \in \D(\H_{\sX})$, we have 
\begin{align*}
\hat{V}(\rho) &= \min_{\pi \in \D(\sA)} \langle c + \beta \, \N^{\dag}(\xi^*), \rho \otimes \pi \rangle \\
&= \min_{\pi \in \D(\sA)} \left[\langle c,\rho \otimes \pi \rangle + \beta \, \langle \xi^*, \N(\rho \otimes \pi) \rangle\right] \\
&= \min_{\pi \in \D(\sA)} \left[\langle c,\rho \otimes \pi \rangle + \beta \, \hat{V}(\N(\rho \otimes \pi)) \right] \eqqcolon \L \hat{V}(\rho).
\end{align*}
Hence $\hat{V}$ is a fixed point of the operator $\L$. Since $\L$ has a unique fixed point $V^*$ by Theorem~\ref{dynamic-qmdp}, which is the optimal value function of q-MDP with open-loop policies, we have $V^*(\rho) = \langle \xi^*,\rho \rangle$; that is, $V^*$ is linear in $\rho$.

Note that for any $\rho_0$, we have  
\begin{align*}
(1-\beta) \, V^*(\rho_0) = \langle \xi^*,(1-\beta) \, \rho_0 \rangle \geq (\textbf{SDP}) \geq (\textbf{SDP}^{\dag}).
\end{align*} 
Since $\xi^*$ is a feasible solution for $(\textbf{SDP}^{\dag})$, $\xi^*$ is indeed the optimal solution in view of above equation. It should be noted that $\xi^*$ is indeed the optimal solution for $(\textbf{SDP}^{\dag})$ regardless of the initial density operator $\rho_0$, meaning that for any $\rho_0$, $\xi^*$ is a common optimal solution for $(\textbf{SDP}^{\dag})$. Hence, for any $\rho_0$, we have
\begin{align*}
(1-\beta) \, V^*(\rho_0) = \langle \xi^*,(1-\beta) \, \rho_0 \rangle = (\textbf{SDP}) = (\textbf{SDP}^{\dag}).
\end{align*}

Now it is time to prove the existence of optimal stationary open-loop policy for q-MDP.

\begin{theorem}\label{quantum-stationary}
Under Assumption~\ref{as1}, the q-MDP has a stationary optimal open-loop policy for any initial density operator $\rho_0$.
\end{theorem}

\begin{proof}
To this end, define the following set-valued map $\rR:\D(\H_{\sA}) \rightarrow 2^{\D(\H_{\sA})}$ as follows. If $\pi \in \D(\H_{\sA})$ is the input, then $\{\rho_t^{\bgamma}\}$ denotes the state sequence under the stationary open-loop policy $\bgamma = \{\gamma\}$, where $\gamma(\rho) = \rho \otimes \pi$; that is
$
\rho_{t+1}^{\bgamma} = \N(\rho_t^{\bgamma} \otimes \pi), \,\, \forall t\geq0. 
$
Define the state occupation density operator 
$\rho^{\bgamma} \coloneqq \sum_{t=0}^{\infty} (1-\beta) \, \beta^t \, \rho_t^{\bgamma} \eqqcolon \sum_{t=0}^{\infty} \lambda_t \, \rho_t^{\bgamma}.$ 
Then, the set $\rR(\pi)$ is defined as 
$$
\rR(\pi) \coloneqq \argmin_{\hat{\pi} \in \D(\H_{\sA})} \left[\langle c,\rho^{\bgamma} \otimes \hat{\pi} \rangle + \beta \, \langle \xi^*, \N(\rho^{\bgamma} \otimes \hat{\pi}) \rangle\right].
$$
First of all, it is straightforward to prove that $\rR(\pi)$ is non-empty, closed, and convex for any $\pi \in \D(\H_{\sA})$. Moreover, if $\hat{\pi}_{\rho^{\bgamma}} \in \rR(\pi)$, then we have 
\begin{align*}
&\langle \xi^*,\rho^{\bgamma} \rangle = \sum_{t=0}^{\infty} \lambda_t \, \langle \xi^*,\rho_t^{\bgamma} \rangle = \sum_{t=0}^{\infty} \lambda_t \, V^*(\rho_t^{\bgamma})  \\
&= \sum_{t=0}^{\infty} \lambda_t \, \min_{\hat{\pi} \in \D(\H_{\sA})} \left[\langle c,\rho_t^{\bgamma} \otimes \hat{\pi} \rangle + \beta \, \langle \xi^*, \N(\rho_t^{\bgamma} \otimes \hat{\pi}) \rangle\right] \,\, \text{(dynamic programming principle)} \\
&= \left[\langle c,\rho^{\bgamma} \otimes \hat{\pi}_{\rho^{\bgamma}} \rangle + \beta \, \langle \xi^*, \N(\rho^{\bgamma} \otimes \hat{\pi}_{\rho^{\bgamma}}) \rangle\right] \,\, \text{(dynamic programming principle)} \\
&= \sum_{t=0}^{\infty} \lambda_t \, \left[\langle c,\rho_t^{\bgamma} \otimes \hat{\pi}_{\rho^{\bgamma}} \rangle + \beta \, \langle \xi^*, \N(\rho_t^{\bgamma} \otimes \hat{\pi}_{\rho^{\bgamma}}) \rangle\right] \\
&\geq \sum_{t=0}^{\infty} \lambda_t \, \min_{\hat{\pi} \in \D(\H_{\sA})} \left[\langle c,\rho_t^{\bgamma} \otimes \hat{\pi} \rangle + \beta \, \langle \xi^*, \N(\rho_t^{\bgamma} \otimes \hat{\pi}) \rangle\right].
\end{align*}
Since $\lambda_t >0$ for all $t$, we have
$
\hat{\pi}_{\rho^{\bgamma}} \in \argmin_{\hat{\pi} \in \D(\H_{\sA})} \left[\langle c,\rho_t^{\bgamma} \otimes \hat{\pi} \rangle + \beta \, \langle \xi^*, \N(\rho_t^{\bgamma} \otimes \hat{\pi}) \rangle\right]
$
for all $t$. 
Hence, 
$
\rR(\pi) \subset \argmin_{\hat{\pi} \in \D(\H_{\sA})} \left[\langle c,\rho_t^{\bgamma} \otimes \hat{\pi} \rangle + \beta \, \langle \xi^*, \N(\rho_t^{\bgamma} \otimes \hat{\pi}) \rangle\right]
$
for all $t$. Now, we need the following intermediate result to complete the proof.

\begin{proposition}\label{stationary-qmdp}
There exists a fixed point $\pi^*$ of $\rR$; that is, $\pi^* \in \rR(\pi^*)$.
\end{proposition}

\begin{proof}
This can be established via Kakutani's fixed point theorem \cite[Corollary 17.55]{AlBo06}. Firstly, it is known that $\rR(\pi)$ is non-empty, closed, and convex for any $\pi \in \D(\H_{\sA})$. Hence it is sufficient to prove that $\rR$ has a closed graph. Suppose that $(\pi_n,\xi_n) \rightarrow (\pi,\xi)$ as $n \rightarrow \infty$, where $\pi_n \in \D(\H_{\sA})$ and $\xi_n \in \rR(\pi_n)$; that is
$
\|\pi_n-\pi\|_{HS} + \|\xi_n-\xi\|_{HS} \rightarrow 0. 
$
Here $\|\cdot\|_{HS}$ is the Hilbert-Schmidt norm \cite[Definition 4.4.5]{Wil13}. We want to prove that $\xi \in \rR(\pi)$, which implies the closedness of the graph of $\rR$ and completes the proof via Kakutani's fixed point theorem. 

Let us first consider state occupation density operators under stationary open-loop policies induced by $\pi_n$ and $\pi$
\begin{align*}
\rho^{\pi_n} = (1-\beta) \, \sum_{t=0}^{\infty} \beta^t \, \rho_t^{\pi_n}, \,\,
\rho^{\pi} = (1-\beta) \, \sum_{t=0}^{\infty} \beta^t \, \rho_t^{\pi}.
\end{align*}
We first prove that $\rho_t^{\pi_n} \rightarrow \rho_t^{\pi}$ for all $t$, which implies that $\rho^{\pi_n} \rightarrow \rho^{\pi}$. Indeed, the claim is true for $t=0$ as $\rho_0^{\pi_n} = \rho_0^{\pi} = \rho_0$. Suppose it is true for some $t\geq0$ and consider $t+1$. We have 
$$
\rho_{t+1}^{\pi_n} = \N(\rho_{t}^{\pi_n}\otimes\pi_n) \,\, \text{and} \,\, \rho_{t+1}^{\pi} = \N(\rho_{t}^{\pi}\otimes\pi).
$$
Suppose that $\N$ has the following Kraus representation $\N(\sigma) = \sum_{l \in L} K_l \sigma K_l^{\dag}$. Then we have
\begin{align*}
\|\N(\rho_{t}^{\pi_n}\otimes\pi_n) -\N(\rho_{t}^{\pi}\otimes\pi) \|_{HS} &=  \|\sum_{l \in L} K_l (\rho_{t}^{\pi_n}\otimes\pi_n - \rho_{t}^{\pi}\otimes\pi) K_l^{\dag} \|_{HS} \\
&\leq  \sum_{l \in L} \| K_l (\rho_{t}^{\pi_n}\otimes\pi_n - \rho_{t}^{\pi}\otimes\pi) K_l^{\dag} \|_{HS} \\
&\leq  \sum_{l \in L} \| K_l \|_{HS} \|(\rho_{t}^{\pi_n}\otimes\pi_n - \rho_{t}^{\pi}\otimes\pi)\|_{HS} \| K_l^{\dag} \|_{HS} \\
&\leq  \sum_{l \in L} \| K_l \|_{HS} \|\rho_{t}^{\pi_n} -\rho_{t}^{\pi} \|_{HS} \|\pi_n - \pi\|_{HS} \| K_l^{\dag} \|_{HS} \\
&\rightarrow 0 \,\, \text{as} \,\, n \rightarrow \infty.
\end{align*}
Hence, by induction, $\rho_t^{\pi_n} \rightarrow \rho_t^{\pi}$ for all $t$, and so, $\rho^{\pi_n} \rightarrow \rho^{\pi}$. 

Now define the function $f:\D(\H_{\sX}) \rightarrow \R$ as follows
$
f(\rho) \coloneqq \min_{\hat{\pi} \in \D(\H_{\sA})} \left[ \langle c,\rho\otimes \hat{\pi} \rangle + \beta \, \langle \xi^*,\N(\rho\otimes \hat{\pi}) \right].
$
Since the function 
$(\rho,\hat{\pi}) \mapsto \left[ \langle c,\rho\otimes \hat{\pi} \rangle + \beta \, \langle \xi^*,\N(\rho\otimes \hat{\pi}) \right]$
is jointly continuous and $\D(\H_{\sA}), \D(\H_{\sX})$ are compact, the function $f$ is continuous \cite[Proposition D.6]{HeLa96}. Hence, $f(\rho^{\pi_n}) \rightarrow f(\rho^{\pi})$ as  $\rho^{\pi_n} \rightarrow \rho^{\pi}$. But one can prove using Kraus representation of the channel $\N$ that
$$
f(\rho^{\pi_n}) = \left[ \langle c,\rho^{\pi_n}\otimes \xi_n \rangle + \beta \, \langle \xi^*,\N(\rho^{\pi_n}\otimes \xi_n) \right] \rightarrow \left[ \langle c,\rho^{\pi}\otimes \xi \rangle + \beta \, \langle \xi^*,\N(\rho^{\pi}\otimes \xi) \right].
$$
Hence $f(\rho^{\pi}) = \left[ \langle c,\rho^{\pi}\otimes \xi \rangle + \beta \, \langle \xi^*,\N(\rho^{\pi}\otimes \xi) \right]$; that is, $\xi \in \rR(\pi)$, which completes the proof.
\end{proof}

Note that if $\pi^*$ is the fixed point of $\rR(\pi^*)$, then we have the following properties 
\begin{align*}
&\bullet \, \rho_{t+1} = \N(\rho_t \otimes \pi^*), \,\, \forall t\geq0 \\
&\bullet \, \pi^* \in \argmin_{\hat{\pi} \in \D(\H_{\sA})} \left[\langle c,\rho_t \otimes \hat{\pi} \rangle + \beta \, \langle \xi^*, \N(\rho_t \otimes \hat{\pi}) \rangle\right] \,\, \forall t\geq0 
\end{align*}
Hence, by condition (Opt), the stationary open-loop policy $\bgamma^* = \{\gamma^*\}$ with $\gamma^*(\rho) = \rho \otimes \pi^*$ is optimal. This indeed completes the proof of the existence of optimal stationary open-loop policy in view of Proposition~\ref{stationary-qmdp}. 
\end{proof}

\subsection{Computation of Optimal Cost Function and Policy}\label{open-computation}

In the open-loop case, determining $\xi^*$ in Assumption~\ref{as1} is challenging, and so, making it difficult to derive a general expression for $V^*(\cdot) = \langle \xi^*, \cdot \rangle$. However, since $\D(\H_{\sX})$ is compact, we can obtain an approximation of $V^*$. Indeed, let $\D_n(\H_{\sX}) \subset \D(\H_{\sX})$ be a finite $1/n$-net in $\D(\H_{\sX})$ with respect to Hilbert-Schmidt norm; that is,
$$
\sup_{\rho \in \D(\H_{\sX})} \min_{\tilde \rho \in \D_n(\H_{\sX})} \|\rho-\tilde \rho\|_{HS} < \frac{1}{n}.
$$
Let $Q_n:\D(\H_{\sX}) \rightarrow \D_n(\H_{\sX})$ be the corresponding quantization map; that is, for any $\rho \in \D(\H_{\sX})$, we have $\|\rho - \rho_q\|_{HS} < 1/n$, where $\rho_q \coloneqq Q_n(\rho)$.
The following method gives an approximation for $V^*$.

\begin{algorithm}[H] 
\caption{Computing the optimal cost function $V^*(\rho)$ approximately}
\begin{algorithmic}\label{closed-opt-cost}
\FOR{$\tilde \rho \in \D_n(\H_{\sX})$}
\STATE{
Solve $(\textbf{SDP}^{\dag})$ for $\rho_0 = \tilde \rho$ and let $\xi_{\tilde \rho}$ be an optimal solution (not necessarily $\xi^*$).
}
\ENDFOR
\RETURN{$V^*_n(\rho) = \left\langle \xi_{\rho_q},\rho \right \rangle $ for all $\rho \in \D(\H_{\sX})$.}
\end{algorithmic}
\end{algorithm}

\begin{proposition}\label{open-opt-cost-prop}
The method above returns optimal cost function $V^*(\cdot)$ approximately as follows:
$$
\|V^*-V_n^*\|_{\infty} \leq \frac{\|c\|_{HS} \, \sqrt{|\sX|}}{(1-\beta) \, n},
$$
where recall that $|\sX|$ is the dimension of $\H_{\sX}$. 
\end{proposition}

\begin{proof}
For any $\tilde \rho  \in \D_n(\H_{\sX})$, let $\xi_{\tilde \rho}$ be an optimal solution for $(\textbf{SDP}^{\dag})$ for $\rho_0 = \tilde \rho$. Since $\xi^*$ in Assumption~\ref{as1} is also an optimal solution of  $(\textbf{SDP}^{\dag})$ regardless of the initial state, we have 
\begin{align*}
\langle \xi^*,\tilde \rho \rangle = \langle \xi_{\tilde \rho},\tilde \rho \rangle.
\end{align*}
This implies that 
\begin{align*}
\|V^*-V_n^*\|_{\infty} &= \sup_{\rho \in \D(\H_{\sX})} |V^*(\rho)-V_n^*(\rho)| \\
&= \sup_{\rho \in \D(\H_{\sX})} |\langle \xi^*,\rho \rangle - \langle \xi^*,\rho_q \rangle| \\ 
&= \sup_{\rho \in \D(\H_{\sX})} |\langle \xi^*,\rho - \rho_q \rangle| \\ 
&\leq \|\xi^*\|_{HS} \, \|\rho-\rho_q\|_{HS} \\
&\leq \|\xi^*\|_{HS} \, \frac{1}{n}.
\end{align*}
Note that for any $\rho \in \D(\H_{\sX})$, the absolute value of the optimal cost $V^*(\rho)$ is upper bounded by $\|c\|_{HS}/(1-\beta)$. Moreover, $\sup_{\rho \in \D(\H_{\sX})} \langle \xi^*,\rho \rangle$ ($= \sup_{\rho \in \D(\H_{\sX})} V^*(\rho)$) is equal to the spectral norm of $\xi^*$ and Hilbert-Schmidt norm of any Hermitian operator $\xi$ on $\H_{\sX}$ is upper bounded by spectral norm of $\xi$ times $\sqrt{|\sX|}$. Therefore, 
$$
\|\xi^*\|_{HS} \leq \frac{\|c\|_{HS} \, \sqrt{|\sX|}}{(1-\beta)}. 
$$
This completes the proof.
\end{proof}

In the procedure outlined above, we must solve $(\textbf{SDP}^{\dag})$ a number of times equal to the cardinality of $\D_n(\H_{\sX})$. As $(\textbf{SDP}^{\dag})$ is a semi-definite programming problem, it can be efficiently solved using various well-established algorithms (see \cite{WoSaVa00}). Therefore, this method remains computationally practical as long as the cardinality of $\D_n(\H_{\sX})$ is not excessively large.

Now it is time to present a method for computing stationary optimal open-loop policy. Note that the state-action occupation density operator of the optimal stationary open-loop policy can be written in the following form:
$$
\sigma^{\bgamma^*} = (1-\beta) \, \sum_{t=0}^{\infty} \beta^t \, \rho_t^{\bgamma^*} \otimes \pi^* = \rho^{\bgamma^*} \otimes \pi^*.
$$
Since 
$$
(1-\beta) \, V^*(\rho_0) = \langle c, \sigma^{\bgamma^*} \rangle,
$$
the density operator $\sigma^{\bgamma^*}=\rho^{\bgamma^*} \otimes \pi^*$ is an optimal solution of the primal $(\textbf{SDP})$, which is in the product form.
This observation suggests the following characterization of the stationary optimal open-loop policy. 

\begin{proposition}\label{open-optimal-policy-prop}
Let $\sigma^*$ be an optimal solution for $(\textbf{SDP})$, which also takes the product form $\rho^* \otimes \pi^*$. Then, $\bgamma^* = \{\gamma^*\}$ with $\gamma^*(\rho) \coloneqq \rho \otimes \pi^*$ is an optimal stationary open-loop policy. 
\end{proposition}

\begin{proof}
The existence of at least one such viable solution can be deduced from the presence of the aforementioned optimal stationary open-loop policy.

It is important to mention that there may exist multiple optimal solutions in the product form for $(\textbf{SDP})$, and these may not necessarily correspond to the state-action occupation density operator of an optimal stationary open-loop policy. Nevertheless, regardless of the particular optimal solution in product form for $(\textbf{SDP})$, we demonstrate that any such solution leads to an optimal stationary open-loop policy for the q-MDP. Indeed, let $\bgamma^*$ be the stationary open-loop policy in the statement of the proposition. Then, we have 
\begin{align*}
\rho^* &= (1-\beta) \, \rho_0 + \beta \, \N(\rho^*\otimes\pi^*) \\
&= (1-\beta) \, \rho_0 + \beta \, \N(\{(1-\beta) \, \rho_0 + \beta \, \N(\rho^*\otimes\pi^*)\}\otimes\pi^*) \\
&= (1-\beta) \, \rho_0 + (1-\beta) \, \beta \, \N(\rho^*\otimes\pi^*) + \beta^2 \, \N(\N(\rho^*\otimes\pi^*)\otimes\pi^*) \\
&\phantom{x}\vdots \\
&\rightarrow (1-\beta) \, \sum_{t=0}^{\infty} \beta^t \, \rho_t^{\bgamma^*} = \rho^{\bgamma^*}.
\end{align*}
Hence, state-action occupation density operator under the stationary open-loop policy $\bgamma^*$ is $\rho^*\otimes\pi^*$. Since $\rho^*\otimes\pi^*$ is optimal for  $(\textbf{SDP})$, we have  
\begin{align*}
(1-\beta) \, V^*(\rho_0) = (\textbf{SDP}) = \langle c,\rho^*\otimes\pi^* \rangle = (1-\beta) \, V(\bgamma^*,\rho_0).
\end{align*} 
Here, the last equality follows from the fact that $\rho^*\otimes\pi^*$ is the state-action occupation density operator under the stationary open-loop policy $\bgamma^*$. Hence, $\bgamma^*$ is an optimal stationary open-loop policy. 
\end{proof}

Generally, for any initial state $\rho_0$, the semi-definite programming problem $(\textbf{SDP})$ can be efficiently solved using a variety of well-established algorithms (see \cite{WoSaVa00}). However, in Proposition~\ref{open-optimal-policy-prop}, we are required to find an optimal solution to $(\textbf{SDP})$ in product form. As we will formulate below, this problem can be rephrased as a bi-linear optimization task, which is computationally more challenging to solve:

\begin{tcolorbox} 
[colback=white!100]
\begin{align}
(\textbf{BIL}) \text{                         }&\min_{\substack{\sigma \in \cL_H(\H_{\sX}\otimes\H_{\sA}) \\ \pi \in \cL_H(\H_{\sA})}} \text{ } \langle c,\sigma \rangle
\nonumber \\*
&\text{s.t.} \, \T(\sigma) = (1-\beta) \,\rho_0 \,\, \text{and} \,\, \sigma \succcurlyeq 0 \nonumber \\
&\phantom{s.t.} \sigma = \tr_{\sA}(\sigma) \otimes \pi, \,\, \tr(\pi) = 1, \,\, \pi \succcurlyeq 0 \label{eqqqq1} 
\end{align} 
\end{tcolorbox} 
Here, the bi-linear constraint in (\ref{eqqqq1}) ensures that the product form condition in Proposition~\ref{open-optimal-policy-prop} holds. In view of this, we can conclude the following:

\begin{proposition}\label{open-optimal-policy-prop_comp}
Let $(\sigma^*,\pi^*)$ be an optimal solution for $(\textbf{BIL})$. Then, under Assumption~\ref{as1}, the quantum channel 
$$
\gamma^*(\rho) \coloneqq  \rho \otimes \pi^*,
$$
is the optimal stationary open-loop policy.
\end{proposition}

To compute an optimal stationary open-loop policy, we must solve the bi-linear optimization problem $(\textbf{BIL})$. Bi-linear optimization problems are challenging due to their non-convex nature, which results in multiple local optima and complicates the search for a global solution. This difficulty stems from the coupling of variables in bi-linear terms. These problems are often NP-hard, with complexity increasing exponentially as the problem size grows. As a result, solving them typically requires specialized algorithms and significant computational resources, making them more difficult than linear or convex problems.

\section{Algorithms for Classical-state-preserving Closed-loop  Quantum Policies}\label{sec3}

Recall that a \emph{classical-state-preserving closed-loop  policy} is a sequence quantum channels $\bgamma=\{\gamma_{t}\}$ from $\H_{\sX}$ to $\H_{\sX} \otimes \H_{\sA}$ such that, for all $t$, $\tr_{\sA}= \Inv_{\S}(\gamma_t)$; that is, 
$
\tr_{\sA}(\gamma_t(\rho)) = \rho, \,\,\, \forall \rho \in \S,
$
where $\S = \{|x\rangle \langle x| : x \in \sX\}$. We refer to q-MDPs using classical-state-preserving closed-loop  policies as qw-MDPs to distinguish it from q-MDPs with open-loop policies.

In the rest of this section, we take a comparable approach to the one we used for q-MDPs with open-loop policies. Initially, we deduce the dynamic programming principle and subsequently construct a semi-definite programming (SDP) formulation for qw-MDPs. Following that, we establish the link between the dual SDP and the dynamic programming equation, resulting in a linear optimal value function. Leveraging this linearity, we demonstrate the existence of an optimal stationary classical-state-preserving closed-loop  policy.

\subsection{Dynamic Programming for qw-MDP}\label{sub2sub1sec3}

The dynamic programming operator for qw-MDP, denoted as 
$\L_w:C_b(\D(\H_{\sX})) \rightarrow C_b(\D(\H_{\sX})),$
is defined as follows
\begin{align*}
\L_w V(\rho) &\coloneqq \min_{\gamma \in \cC_w} \left[ \langle c,\gamma(\rho) \rangle + \beta \, V(\N \circ \gamma(\rho)) \right]. 
\end{align*}
To begin with, demonstrating that $\L_w$ preserves the continuity of functions is straightforward. This is due to $\cC_w$ being a closed subset of a compact set of quantum channels and the function $\langle c,\gamma(\rho) \rangle + \beta , V(\N \circ \gamma(\rho))$ is continuous in both $\rho$ and $\gamma$ when $V$ is continuous. Furthermore, it can be proven that the operator $\L_w$ is a $\beta$-contraction with respect to the sup-norm on $C_b(\D(\H_{\sX}))$. As a result, it possesses a unique fixed point denoted as $V_{\mathrm{fw}} \in C_b(\D(\H_{\sX}))$.

Next, we establish a connection between the optimal value function $V^*_{\mathrm{w}}(\rho_0)$ of qw-MDP, defined as 
$\min_{\bgamma \in \Gamma_w} V(\bgamma,\rho_0),$
and $V_{\mathrm{fw}}$. This is indeed very similar to the open-loop case.

\begin{theorem}\label{dynamic-qwmdp}
We have $V^*_{\mathrm{w}}(\rho_0) = V_{\mathrm{fw}}(\rho_0)$ for all $\rho_0 \in \D(\H_{\sX})$; that is, the unique fixed point of $\L_w$ is $V^*_{\mathrm{w}}$.
\end{theorem}

\begin{proof}
The proof is very similar to the open-loop case but for completeness, we give the full proof. Fix any closed-loop policy $\bgamma = \{\gamma_t\}_{t\geq0} \in \Gamma_w$. Then, the cost function of $\bgamma$ satisfies the following 
\begin{align*}
V(\bgamma,\rho_0) &= \langle c,\gamma_0(\rho_0) \rangle + \beta \, V(\{\gamma_t\}_{t\geq 1},\N \circ \gamma_0(\rho_0)) \\
&\geq \min_{\gamma_0 \in \cC_w} \left[ \langle c,\gamma_0(\rho_0) \rangle + \beta \, V(\{\gamma_t\}_{t\geq 1},\N \circ \gamma_0(\rho_0)) \right] \\
&\geq \min_{\gamma_0 \in \cC_w} \left[ \langle c,\gamma_0(\rho_0) \rangle + \beta \, V^*_{\mathrm{w}}(\N \circ \gamma_0(\rho_0)) \right] = \L_w V^*(\rho_0).
\end{align*}
Since above inequality is true for any closed-loop policy $\bgamma$, we have $V^*_{\mathrm{w}} \geq \L V^*_{\mathrm{w}}$. Using the latter inequality and the monotonicity of $\L_w$, one can also prove that $V^*_{\mathrm{w}} \geq \L^n_wV^*_{\mathrm{w}}$ for any positive integer $n$. By Banach fixed point theorem, it is known that $\L^n_w V^*_{\mathrm{w}} \rightarrow V_{\mathrm{fw}}$ in sup-norm as $n\rightarrow \infty$. Hence, $V^*_{\mathrm{w}} \geq  V_{\mathrm{fw}}$. 

To prove the converse, for any $\rho \in \D(\H_{\sX})$, let 
$$
\gamma_{\rho} \in \argmin_{\gamma \in \cC_w} \left[ \langle c,\gamma(\rho) \rangle + \beta \, V_{\mathrm{fw}}(\N\circ\gamma(\rho)) \right].
$$
The existence of $\gamma_{\rho}$ for any $\rho$ follows from the facts that $\cC_w$ is a closed subset (as the constraint in Definition~\ref{closed1} is preserved under the convergence of quantum channels) of compact set of quantum channels from $\H_{\sX}$ to $\H_{\sX}\otimes\H_{\sA}$ \cite[Proposition 2.28]{Wat18} and the function $\langle c,\gamma(\rho) \rangle + \beta \, V(\N \circ \gamma(\rho))$ is continuous in $\rho$ and $\gamma$. With this construction, we then have
\begin{align*}
V_{\mathrm{fw}}(\rho_0) &=  \langle c,\gamma_{\rho_0}(\rho_0)\rangle + \beta \, V_{\mathrm{fw}}(\N \circ \gamma_{\rho_0}(\rho_0)) \\
&\eqqcolon \langle c,\sigma_0  \rangle + \beta \, V_{\mathrm{fw}}(\rho_1) \\
&= \langle c,\sigma_0  \rangle + \beta \, \left[\langle c,\gamma_{\rho_1}(\rho_1) \rangle + \beta \, V_{\mathrm{fw}}(\N \circ \gamma_{\rho_1}(\rho_1)) \right] \\
&\eqqcolon \langle c,\sigma_0  \rangle + \beta \, \left[\langle c,\sigma_1 \rangle + \beta \, V_{\mathrm{fw}}(\rho_2) \right] \\
&\phantom{x}\vdots \\
&= \sum_{t=0}^{N-1} \beta^t \, \langle c,\sigma_t \rangle + \beta^N \, V_{\mathrm{fw}}(\rho_N) \rightarrow V(\{\gamma_t\},\rho_0) \,\, \text{as} \,\ N\rightarrow\infty
\end{align*}
where $\gamma_t(\rho) \coloneqq  \gamma_{\rho_t}(\rho)$ for each $t$. Hence, $V_{\mathrm{fw}} \geq V^*_{\mathrm{w}}$. This completes the proof.
\end{proof}

In view of the proof above, we can conclude that the optimal classical-state-preserving closed-loop  policy $\bgamma^*_{\mathrm{w}}$ for any given initial density operator $\rho_0$ can be constructed as follows: $\gamma^*_{0,\mathrm{w}}(\cdot) = \gamma_{\rho_0}(\cdot)$ and $\gamma^*_{t+1,\mathrm{w}}(\cdot) = \gamma_{\rho_{t+1}}(\cdot)$ for $t\geq0$, where $\rho_{t+1} = \N \circ \gamma_{\rho_t}(\rho_t)$. This leads to the following optimality condition for qw-MDP:

\begin{tcolorbox} 
[colback=white!100]
\begin{itemize}
\item[ ] \hspace{-20pt} {\bf (Opt-w):} 
If the classical-state-preserving closed-loop  policy $\bgamma = \{\gamma_t\}$ satisfies the following conditions
\begin{align*}
&\bullet \, \rho_{t+1} = \N \circ \gamma_t(\rho_t), \,\, \forall t\geq0 \\
&\bullet \, \gamma_t \in \argmin_{\gamma \in \cC_w} \left[ \langle c,\gamma(\rho_t)\rangle + \beta \, V^*_{\mathrm{w}}(\N \circ \gamma(\rho_t)) \right]
\end{align*}
then $\bgamma$ is optimal.
\end{itemize}
\end{tcolorbox}

This implies that knowledge of the optimal value function $V^*_{\mathrm{w}}$ is necessary to determine the optimal classical-state-preserving closed-loop  policy for $\rho_0$. It is also observed that the optimal classical-state-preserving closed-loop  policy depends on the initial distribution and changes over time (i.e., non-stationary).

In the following section, we will present a semidefinite programming formulation of the qw-MDP, which demonstrates that the optimal value function $V^*_{\mathrm{w}}(\rho)$ can be represented as a linear function of $\rho$. Additionally, we will show that this formulation guarantees the existence of an optimal stationary classical-state-preserving closed-loop policy.
 
\subsection{Semi-Definite Programming Formulation of qw-MDP}\label{sub3sub1sec3}

Consider the qw-MDP under $\beta$-discounted cost criterion with initial density operator $\rho_0$. Given any classical-state-preserving closed-loop  policy $\bgamma$, define the state-action occupation density operator as follows
$$
\sigma^{\bgamma} \coloneqq (1-\beta) \, \sum_{t=0}^{\infty} \beta^t \, \sigma_t.
$$
Recall that one can write the normalized cost of $\bgamma$ in the following form
$$
(1-\beta) \, V(\bgamma,\rho_0) =  C(\sigma^{\bgamma}) = \langle c,\sigma^{\bgamma} \rangle
$$
as a result of linearity of $C$. Let $\N_{qc}:\D(\H_{\sX}) \rightarrow \D(\H_{\sX})$ be a quantum-to-classical channel that is defined as follows \cite[Definition 4.6.7]{Wil13}
$$
\N_{qc}(\rho) \coloneqq \sum_{x \in \sX} \tr(\rho |x\rangle \langle x|) \, |x\rangle \langle x|.
$$
This channel transforms any quantum state (density operator) via measurement operators $\{|x\rangle \langle x|\}$ into a classical state (probability measure). 

\smallskip

\begin{proposition}\label{importantidentity}
Let $\gamma \in \cC_w$. Then, for any $\rho \in \D(\H_{\sX})$, we have
$$
\N_{qc} \circ \tr_{\sA} \circ \gamma(\rho) = \N_{qc}(\rho).
$$ 
\end{proposition}
\begin{proof}
By Proposition~\ref{structure-qwmdp}, we have 
\begin{align*}
\gamma(\rho) = \sum_{x,y \in \sX} \tr(\rho |y\rangle \langle x|) \, |x\rangle \langle y| \otimes \left(\sum_{a,b \in \sA} \langle \phi_{y,b},\phi_{x,a} \rangle \, |a\rangle \langle b | \right)
\end{align*}
for some collection of vectors $\{|\phi_{x,a} \rangle\}_{(x,a)\in\sX\times\sA}$ in some Hilbert space $\H_L$ with $\dim(\H_L) \leq |\sX|^2|\sA|$ such that $\sum_{a \in \sA} \langle \phi_{x,a},\phi_{x,a} \rangle=1$ for each $x \in \sX$.
Then we have 
\begin{align*}
\N_{qc} \circ \tr_{\sA} \circ \gamma(\rho) &= \N_{qc} \left( \sum_{x,y \in \sX} \tr(\rho |y\rangle \langle x|) \, |x\rangle \langle y| \otimes \left(\sum_{a,b \in \sA} \langle \phi_{y,b},\phi_{x,a} \rangle \, \tr(|a\rangle \langle b |) \right) \right) \\
&= \N_{qc} \left( \sum_{x,y \in \sX} \tr(\rho |y\rangle \langle x|) \, |x\rangle \langle y| \cdot \left(\sum_{a \in \sA} \langle \phi_{y,a},\phi_{x,a} \rangle \right) \right) \\
&=  \sum_{x,y \in \sX} \tr(\rho |y\rangle \langle x|) \, \N_{qc}(|x\rangle \langle y|) \cdot \left(\sum_{a \in \sA} \langle \phi_{y,a},\phi_{x,a} \rangle \right) \\
&=  \sum_{x \in \sX} \tr(\rho |x\rangle \langle x|) \, |x\rangle \langle x| \cdot \underbrace{\left(\sum_{a \in \sA} \langle \phi_{x,a},\phi_{x,a} \rangle \right)}_{=1} \\
&= \N_{qc}(\rho).
\end{align*}
This completes the proof. 
\end{proof}

The identity in the proposition above is crucial and will be used frequently throughout this section. Using this identity, we can proceed with the following computation:
\begin{align}
&\N_{qc} \circ \tr_{\sA}(\sigma^{\bgamma}) = (1-\beta) \, \sum_{t=0}^{\infty} \beta^t \, \N_{qc} \circ \tr_{\sA}(\sigma_t) \,\, (\text{by linearity of $\N_{qc} \circ \tr_{\sA}$})\nonumber \\
&= (1-\beta) \, \sum_{t=0}^{\infty} \beta^t \, \N_{qc} \circ \tr_{\sA}\circ \gamma_t(\rho_t) \nonumber \\
&= (1-\beta) \, \sum_{t=0}^{\infty} \beta^t \N_{qc}(\rho_t) \,\, \text{(by above identity)} \nonumber \\
&= (1-\beta) \,\N_{qc}(\rho_0) + (1-\beta) \, \beta \sum_{t=1}^{\infty} \beta^{t-1} \, \N_{qc} \circ \N(\sigma_{t-1}) \nonumber \\
&= (1-\beta) \, \N_{qc}(\rho_0) + \beta \, \N_{qc} \circ \N\left((1-\beta) \,\sum_{t=1}^{\infty} \beta^{t-1} \, \sigma_{t-1}\right)  \,\, \text{(by linearity of $\N_{qc} \circ \N$)}\nonumber \\
&= (1-\beta) \,\N_{qc}(\rho_0) + \beta \, \N_{qc} \circ \N(\sigma^{\bgamma}). \nonumber 
\end{align}
The only distinction in this calculation compared to the one we performed for q-MDPs with open-loop policies is the inclusion of the quantum-to-classical channel $\N_{qc}$. This is due to the fact that the classical-state-preserving closed-loop  policies satisfy only the relaxed reversibility condition.

Above computation suggests the definition of the Hermicity preserving super-operator $\T_w: \cL_H(\H_{\sX}\otimes\H_{\sA}) \rightarrow \cL_H(\H_{\sX})$ as follows
$$
\T_w(\sigma) \coloneqq \N_{qc} \circ \tr_{\sA}(\sigma) - \beta \, \N_{qc} \circ \N(\sigma).
$$
Therefore, $\T_w(\sigma^{\bgamma}) = (1-\beta) \,\N_{qc}(\rho_0)$ for any classical-state-preserving closed-loop  policy $\bgamma$. This observation leads to the following semi-definite programming formulation: 
\begin{tcolorbox} 
[colback=white!100]
\begin{align}
(\textbf{SDP-w}) \text{                         }&\min_{\sigma \in \cL_H(\H_{\sX}\otimes\H_{\sA})} \text{ } \langle c,\sigma \rangle
\nonumber \\*
&\text{s.t.} \, \T_w(\sigma) = (1-\beta) \,\N_{qc}(\rho_0) \,\, \text{and} \,\, \sigma \succcurlyeq 0. \nonumber 
\end{align} 
\end{tcolorbox} 
Note that if $\T_w(\sigma) = (1-\beta) \,\N_{qc}(\rho_0)$ and $\sigma \succcurlyeq 0$, then necessarily $\sigma$ is a density operator. As a result, it is unnecessary to impose a prior restriction on $\sigma$ to be confined within $\D(\H_{\sX}\otimes\H_{\sA})$ in the semi-definite programming formulation. 

Similar to the open-loop case, it is not possible to directly establish the equivalence between $(\textbf{SDP-w})$ and qw-MDP. The key reason is that Proposition~\ref{importantidentity} can be viewed as a relaxation of Proposition~\ref{structure-qwmdp}. Consequently, for any classical-state-preserving closed-loop policy $\bgamma$, there exists a feasible $\sigma^{\bgamma} \in \cL_H(\H_{\sX} \otimes \H_{\sA})$ for $(\textbf{SDP-w})$ such that  
\[
(1-\beta) \, V(\bgamma, \rho_0) = \langle c, \sigma^{\bgamma} \rangle.
\]  
However, the converse does not necessarily hold: if $\sigma \in \cL_H(\H_{\sX} \otimes \H_{\sA})$ is feasible for $(\textbf{SDP-w})$, it is not guaranteed that there exists a classical-state-preserving closed-loop policy $\bgamma^{\sigma}$ satisfying  
\[
(1-\beta) \, V(\bgamma^{\sigma}, \rho_0) = \langle c, \sigma \rangle.
\]  
This follows from the fact that Proposition~\ref{importantidentity} provides a strict relaxation of Proposition~\ref{structure-qwmdp}.

Hence, at present, $(\textbf{SDP-w})$ only acts as a lower bound for qw-MDP:
\begin{align*}
\inf_{\bgamma \in \Gamma_w} (1-\beta) \, V(\bgamma,\rho_0) \geq (\textbf{SDP-w}) \text{                         }&\min_{\sigma \in \cL_H(\H_{\sX}\otimes\H_{\sA})} \text{ } \langle c,\sigma \rangle
\nonumber \\*
&\text{s.t.} \, \T_w(\sigma) = (1-\beta) \,\N_{qc}(\rho_0) \,\, \text{and} \,\, \sigma \succcurlyeq 0.
\end{align*} 
However, later in this section, by employing the dual of $(\textbf{SDP-w})$ and introducing an additional assumption, we can establish their equality. Additionally, we can demonstrate that the optimal value function $V^*_w$ is linear and that a stationary optimal classical-state-preserving closed-loop  policy exists.

To this end, let us obtain the dual of the above SDP. Indeed, we again use Sion min-max theorem:
\begin{align*} 
(\textbf{SDP-w})
& = \min_{\sigma \in \cL_H(\H_{\sX}\otimes\H_{\sA})} \text{ } \langle c,\sigma \rangle
\,\, \text{s.t.} \, \T_w(\sigma) = (1-\beta) \,\N_{qc}(\rho_0) \,\, \text{and} \,\, \sigma \succcurlyeq 0 \nonumber \\
&=\min_{\sigma \succcurlyeq 0} \, \max_{\xi \in \cL_H(\H_{\sX})} \, \langle c,\sigma \rangle  + \langle \xi, (1-\beta) \,\N_{qc}(\rho_0) -\T_w(\sigma) \rangle \\
&\text{(by Sion min-max theorem)} \\
&= \max_{\xi \in \cL_H(\H_{\sX})} \, \min_{\sigma \succcurlyeq 0} \, \langle c,\sigma \rangle  + \langle \xi, (1-\beta) \,\N_{qc}(\rho_0) -\T_w(\sigma) \rangle \\
&= \max_{\xi \in \cL_H(\H_{\sX})} \, \min_{\sigma \succcurlyeq 0} \, \langle \xi,(1-\beta) \, \N_{qc}(\rho_0) \rangle  + \langle c-{\T}^{\dag}_w(\xi), \sigma \rangle \\
&=\max_{\xi \in \cL_H(\H_{\sX})} \, \langle \xi,(1-\beta) \, \N_{qc}(\rho_0) \rangle 
\,\, \text{s.t.} \, c-{\T}^{\dag}_w(\xi)  \succcurlyeq  0 \nonumber \\
&\eqqcolon (\textbf{SDP-w}^{\dag}), 
\end{align*} 
where ${\T}^{\dag}_w:\cL_H(\H_{\sX}) \rightarrow \cL_H(\H_{\sX}\otimes\H_{\sA})$ is the adjoint of $\T_w$ and is given by
$$
{\T}^{\dag}_w(\xi) = \N_{qc}(\xi) \otimes \Id - \beta \, \N^{\dag} \circ \N_{qc}(\xi).
$$
For the dual SDP, we have the following equivalent formulations:
\small
\begin{tcolorbox} 
[colback=white!100]
\begin{align*}
&(\textbf{SDP-w}^{\dag}) \\
&=\max_{\xi \in \cL_H(\H_{\sX})} \text{ }  \langle \xi,(1-\beta) \, \N_{qc}(\rho_0) \rangle 
\,\, \text{s.t.}  \, c + \beta \, \N^{\dag} \circ \N_{qc}(\xi) \succcurlyeq \N_{qc}(\xi) \otimes \Id \\
&=\max_{\xi \in \cL_H(\H_{\sX})} \langle \xi,(1-\beta) \, \N_{qc}(\rho_0) \rangle 
\,\, \text{s.t.}  \, \langle c + \beta \, \N^{\dag} \circ \N_{qc}(\xi), |\psi\rangle \langle \psi | \rangle \geq \langle \N_{qc}(\xi)\otimes \Id,|\psi\rangle \langle \psi | \rangle \\
&\phantom{xxxxxxxxxxxxxxxxxxxxxxxxxxxxxxxxx} \text{for all pure states} \, |\psi \rangle \in \H_{\sX}\otimes\H_{\sA} \\
&=\max_{\xi \in \cL_H(\H_{\sX})} \text{ }  \langle \xi,(1-\beta) \, \N_{qc}(\rho_0) \rangle 
\,\, \text{s.t.}  \, \langle c + \beta \, \N^{\dag} \circ \N_{qc}(\xi), \sigma \rangle \geq \langle \N_{qc}(\xi)\otimes \Id,\sigma \rangle \\
&\phantom{xxxxxxxxxxxxxxxxxxxxxxxxxxxxxxxxx} \text{for all states} \, \sigma \in \D(\H_{\sX}\otimes\H_{\sA}).
\end{align*} 
\end{tcolorbox} 
\normalsize
We now introduce our additional assumption. In the appendix, we establish criteria that are necessary and sufficient to confirm the validity of this assumption.

\smallskip

\begin{assumption}\label{as2}
There exists a feasible $\xi^* \in \cL_H(\H_{\sX})$ for $(\textbf{SDP-w}^{\dag})$ such that for any $\rho \in \D(\H_{\sX})$, we have 
$
\min_{\gamma \in \cC_w} \langle c + \beta \, \N^{\dag} \circ \N_{qc}(\xi^*), \gamma(\rho) \rangle = \langle \N_{qc}(\xi^*),\rho \rangle.
$
\end{assumption}

\smallskip

\begin{remark}
Assumption~\ref{as2} is not overly restrictive. Consider any $\rho \in \D(\H_{\sX})$. If $\xi$ is a feasible solution for $(\textbf{SDP-w}^{\dag})$, it automatically satisfies
\[
\langle c + \beta \, \N^{\dag} \circ \N_{qc}(\xi), \gamma(\rho) \rangle \geq \langle \N_{qc}(\xi)\otimes \Id,\gamma(\rho) \rangle
\]
for any $\gamma \in \cC_w$. Note that
\begin{align*}
\langle \N_{qc}(\xi)\otimes \Id,\gamma(\rho) \rangle &= \langle \N_{qc}(\xi),\tr_{\sA} \circ \gamma(\rho) \rangle \\
&= \langle \xi,\N_{qc} \circ \tr_{\sA} \circ \gamma(\rho) \rangle \quad \text{(since $\N_{qc}^{\dag} = \N_{qc}$)} \\
&= \langle \xi,\N_{qc}(\rho) \rangle \quad \text{(by Proposition~\ref{importantidentity})} \\
&= \langle \N_{qc}(\xi),\rho \rangle \quad \text{(since $\N_{qc}^{\dag} = \N_{qc}$).}
\end{align*}
Therefore, the feasibility condition for $(\textbf{SDP-w}^{\dag})$ can be rewritten as:
\[
\langle c + \beta \, \N^{\dag} \circ \N_{qc}(\xi), \gamma(\rho) \rangle \geq \langle \N_{qc}(\xi),\rho \rangle.
\]
In Assumption~\ref{as2}, we require the existence of a feasible solution $\xi^*$ for $(\textbf{SDP-w}^{\dag})$ that achieves equality in this inequality for any $\rho$, when the left side is minimized over $\gamma \in \cC_w$. 
\end{remark}

\smallskip

\begin{remark}
Under Assumption~\ref{as2}, the optimal value function is provably linear in $\rho$, and the semidefinite programming formulation (\textbf{SDP-w}) is equivalent to the qw-MDP. This equivalence ensures the existence of a stationary optimal classical-state-preserving closed-loop policy. In the absence of Assumption~\ref{as2}, the (\textbf{SDP-w}) formulation remains valuable, providing a lower bound for the qw-MDP problem.
\end{remark}

Let us define a linear function $\hat{V}_w(\rho) \coloneqq \langle \N_{qc}(\xi^*),\rho \rangle$. Then, by Assumption~\ref{as2}, for any $\rho \in \D(\H_{\sX})$, we have 
\begin{align*}
\hat{V}_w(\rho) &= \min_{\gamma \in \cC_w} \langle c + \beta \, \N^{\dag} \circ \N_{qc}(\xi^*), \gamma(\rho) \rangle \\
&= \min_{\gamma \in \cC_w} \left[\langle c,\gamma(\rho) \rangle + \beta \, \langle \N_{qc}(\xi^*), \N \circ \gamma(\rho) \rangle\right] \\
&= \min_{\gamma \in \cC_w} \left[\langle c,\gamma(\rho) \rangle + \beta \, \hat{V}_w(\N \circ \gamma(\rho)) \right] \eqqcolon \L_w \hat{V}_w(\rho).
\end{align*}
Hence $\hat{V}_w$ is a fixed point of the operator $\L_w$. Since $\L_w$ has a unique fixed point $V^*_w$, which is the optimal value function of qw-MDP, the optimal value function $V^*_w(\rho) = \langle \N_{qc}(\xi^*),\rho \rangle$ is linear in $\rho$.

Note that for any $\rho_0$, we have  
\begin{align*}
(1-\beta) \, V^*_w(\rho_0) &= \langle \N_{qc}(\xi^*),(1-\beta) \, \rho_0 \rangle = \langle \xi^*,(1-\beta) \, \N_{qc}(\rho_0) \rangle \\
&\geq (\textbf{SDP-w}) \geq (\textbf{SDP-w}^{\dag}),
\end{align*} 
where recall that the adjoint of $\N_{qc}$ is itself. Since $\xi^*$ is a feasible point for $(\textbf{SDP-w}^{\dag})$, in view of above inequality, $\xi^*$ is an optimal solution for $(\textbf{SDP-w}^{\dag})$. It is important to emphasize that $\xi^*$ serves as the optimal solution for $(\textbf{SDP-w}^{\dag})$, irrespective of the initial density operator $\rho_0$. Hence, for any $\rho_0$, we have
\begin{align*}
(1-\beta) \, V^*_w(\rho_0) = \langle \xi^*,(1-\beta) \, \N_{qc}(\rho_0) \rangle = (\textbf{SDP-w}) = (\textbf{SDP-w}^{\dag}).
\end{align*} 

With the above discussion in mind, we can now prove the existence of an optimal stationary classical-state-preserving closed-loop policy for any initial density operator $\rho_0$.

\begin{theorem}\label{closed-quantum-stationary}
The qw-MDP has a stationary optimal classical-state-preserving closed-loop  policy for any initial density operator $\rho_0$.
\end{theorem}

\begin{proof}
The proof closely follows the proof of Theorem~\ref{quantum-stationary}. However, there are some differences in notation, so we present the complete proof here.

Define the following set-valued map $\rR_w:\cC_w \rightarrow 2^{\cC_w}$ as follows. If $\gamma \in \cC_w$ is the input, then $\{\rho_t^{\bgamma}\}$ is the state sequence under the stationary classical-state-preserving closed-loop  policy $\bgamma = \{\gamma\}$; that is
$
\rho_{t+1}^{\bgamma} = \N \circ \gamma(\rho_t^{\bgamma}), \,\, \forall t\geq0. 
$
Define the state occupation density operator 
$\rho^{\bgamma} \coloneqq \sum_{t=0}^{\infty} (1-\beta) \, \beta^t \, \rho_t^{\bgamma} \eqqcolon \sum_{t=0}^{\infty} \lambda_t \, \rho_t^{\bgamma}.$ 
Then, the set $\rR_w(\gamma)$ is defined as 
$$
\rR_w(\gamma) \coloneqq \argmin_{\tilde{\gamma} \in \cC_w} \left[\langle c,\tilde{\gamma}(\rho^{\bgamma}) \rangle + \beta \, \langle \N_{qc}(\xi^*), \N \circ \tilde{\gamma}(\rho^{\bgamma}) \rangle\right]
$$
First of all, it is straightforward to prove that $\rR_w(\gamma)$ is non-empty, closed, and convex for any $\gamma \in \cC_w$. Moreover, if $\hat{\gamma} \in \rR_w(\gamma)$, then we have 
\small
\begin{align*}
&\langle \N_{qc}(\xi^*),\rho^{\bgamma} \rangle = \left[\langle c,\hat{\gamma}(\rho^{\bgamma}) \rangle + \beta \, \langle \N_{qc}(\xi^*), \N \circ \hat{\gamma}(\rho^{\bgamma}) \rangle\right] \\
&\sum_{t=0}^{\infty} \lambda_t \, \langle \N_{qc}(\xi^*),\rho_t^{\bgamma} \rangle = \sum_{t=0}^{\infty} \lambda_t \, \left[\langle c,\hat{\gamma}(\rho^{\bgamma}_t) \rangle + \beta \, \langle \N_{qc}(\xi^*), \N \circ \hat{\gamma}(\rho^{\bgamma}_t) \rangle\right]   \\
&\sum_{t=0}^{\infty} \lambda_t \, \min_{\tilde{\gamma} \in \cC_w} \left[\langle c,\tilde{\gamma}(\rho_t^{\bgamma}) \rangle + \beta \, \langle \N_{qc}(\xi^*), \N \circ \tilde{\gamma}(\rho_t^{\bgamma}) \rangle\right]  = \sum_{t=0}^{\infty} \lambda_t \, \left[\langle c,\hat{\gamma}(\rho^{\bgamma}_t) \rangle + \beta \, \langle \N_{qc}(\xi^*), \N \circ \hat{\gamma}(\rho^{\bgamma}_t) \rangle\right] \\
&\geq  \sum_{t=0}^{\infty} \lambda_t \, \min_{\tilde{\gamma} \in \cC_w} \left[\langle c,\tilde{\gamma}(\rho_t^{\bgamma}) \rangle + \beta \, \langle \N_{qc}(\xi^*), \N \circ \tilde{\gamma}(\rho_t^{\bgamma}) \rangle\right] \,\, \text{(since $\hat{\gamma}$ is arbitrary)}. 
\end{align*}
\normalsize
Hence, the last inequality is indeed equality. Since $\lambda_t >0$ for all $t$, we have
$$
\hat{\gamma} \in \argmin_{\tilde{\gamma} \in \cC_w} \left[\langle c,\tilde{\gamma}(\rho^{\bgamma}_t) \rangle + \beta \, \langle \N_{qc}(\xi^*), \N \circ \tilde{\gamma}(\rho^{\bgamma}_t) \rangle\right]
$$
for all $t$. 
Hence, 
$$
\rR_w(\gamma) \subset \argmin_{\tilde{\gamma} \in \cC_w} \left[\langle c,\tilde{\gamma}(\rho^{\bgamma}_t) \rangle + \beta \, \langle \N_{qc}(\xi^*), \N \circ \tilde{\gamma}(\rho^{\bgamma}_t) \rangle\right]
$$
for all $t$. 

The following result is the very similar to Proposition~\ref{stationary-qmdp} and will be needed to complete the proof. 

\begin{proposition}\label{stationary-qwmdp}
There exists a fixed point $\gamma^*$ of $\rR_w$.
\end{proposition}

\begin{proof}
The proof of this result is almost the same with the proof of Proposition~\ref{stationary-qmdp}. Firstly, it is known that $\rR_w(\gamma)$ is non-empty, closed, and convex for any $\gamma \in \cC_w$. Hence it is sufficient to prove that $\rR_w$ has a closed graph. Suppose that $(\gamma_n,\xi_n) \rightarrow (\gamma,\xi)$ as $n \rightarrow \infty$, where $\gamma_n \in \cC_w$ and $\xi_n \in \rR_w(\gamma_n)$; that is
$
\|\gamma_n-\gamma\|_{HS} + \|\xi_n-\xi\|_{HS} \rightarrow 0.  
$ 
Here $\|\cdot\|_{HS}$ is the operator norm on super-operators induced by Hilbert-Schmidt norm on Hermitian operators. We want to prove that $\xi \in \rR_w(\gamma)$, which implies the closedness of the graph of $\rR_w$ and completes the proof via Kakutani's fixed point theorem. 

Let us first consider state occupation density operators under stationary classical-state-preserving closed-loop  policies induced by $\gamma_n$ and $\gamma$
\begin{align*}
\rho^{\gamma_n} = (1-\beta) \, \sum_{t=0}^{\infty} \beta^t \, \rho_t^{\gamma_n}, \,\,
\rho^{\gamma} = (1-\beta) \, \sum_{t=0}^{\infty} \beta^t \, \rho_t^{\gamma}.
\end{align*}
We first prove that $\rho_t^{\gamma_n} \rightarrow \rho_t^{\gamma}$ for all $t$, which implies that $\rho^{\gamma_n} \rightarrow \rho^{\gamma}$. Indeed, the claim is true for $t=0$ as $\rho_0^{\gamma_n} = \rho_0^{\gamma} = \rho_0$. Suppose it is true for some $t\geq0$ and consider $t+1$. We have 
$
\rho_{t+1}^{\gamma_n} = \N \circ \gamma_n(\rho_{t}^{\gamma_n}) \,\, \text{and} \,\, \rho_{t+1}^{\pi} = \N \circ \gamma(\rho_{t}^{\gamma}).
$
But if $\N$ has the following Kraus representation $\N(\sigma) = \sum_{l \in L} K_l \sigma K_l^{\dag}$, then we have
\begin{align*}
\|\N \circ \gamma_n(\rho_{t}^{\gamma_n}) -\N \circ \gamma(\rho_{t}^{\gamma}) \|_{HS} &=  \|\sum_{l \in L} K_l (\gamma_n(\rho_{t}^{\gamma_n}) - \gamma(\rho_{t}^{\gamma})) K_l^{\dag} \|_{HS} \\
&\leq  \sum_{l \in L} \| K_l (\gamma_n(\rho_{t}^{\gamma_n}) - \gamma(\rho_{t}^{\gamma})) K_l^{\dag}  \|_{HS} \\
&\leq  \sum_{l \in L} \| K_l \|_{HS} \|\gamma_n(\rho_{t}^{\gamma_n}) - \gamma(\rho_{t}^{\gamma})\|_{HS} \| K_l^{\dag} \|_{HS} \\
&\rightarrow 0 \,\, \text{as} \,\, n \rightarrow \infty.
\end{align*}
Hence, by induction, $\rho_t^{\gamma_n} \rightarrow \rho_t^{\gamma}$ for all $t$, and so, $\rho^{\gamma_n} \rightarrow \rho^{\gamma}$. 

Now define the function $f:\D(\H_{\sX}) \rightarrow \R$ as follows
$$
f(\rho) \coloneqq \min_{\hat{\gamma} \in \cC_w} \left[ \langle c,\hat{\gamma}(\rho) \rangle + \beta \, \langle \N_{qc}(\xi^*),\N \circ \hat{\gamma}(\rho) \right]
$$
One can prove that the function $f$ is continuous. Hence, $f(\rho^{\gamma_n}) \rightarrow f(\rho^{\gamma})$ as  $\rho^{\gamma_n} \rightarrow \rho^{\gamma}$. But one can also prove using Kraus representation of the channel $\N$ that
$$
\hspace{-8pt} f(\rho^{\gamma_n}) = \left[ \langle c,\xi_n(\rho^{\gamma_n}) \rangle + \beta \, \langle \N_{qc}(\xi^*),\N \circ \xi_n(\rho^{\pi_n}) \right] \rightarrow \left[ \langle c,\xi(\rho^{\gamma}) \rangle + \beta \, \langle \N_{qc}(\xi^*),\N \circ \xi(\rho^{\pi}) \right]
$$
Hence $f(\rho^{\gamma}) = \left[ \langle c,\xi(\rho^{\gamma}) \rangle + \beta \, \langle \N_{qc}(\xi^*),\N \circ \xi(\rho^{\gamma}) \right]$; that is, $\xi \in \rR(\gamma)$, which completes the proof.
\end{proof}

Note that if $\gamma^*$ is the fixed point of $\rR_w(\gamma^*)$, then we have the following results 
\begin{align*}
&\bullet \, \rho_{t+1} = \N \circ \gamma^*(\rho_t), \,\, \forall t\geq0 \\
&\bullet \, \gamma^* \in \argmin_{\hat{\gamma} \in \cC_w} \left[\langle c,\hat{\gamma}(\rho_t) \rangle + \beta \, \langle \N_{qc}(\xi^*), \N \circ \hat{\gamma}(\rho_t) \rangle\right] \,\, \forall t\geq0. 
\end{align*}
Therefore, based on condition (Opt-w), the stationary classical-state-preserving closed-loop  policy $\bgamma^* = \{\gamma^*\}$ is optimal. This completes the proof.
\end{proof}

\subsection{Computation of Optimal Cost Function and Policy}\label{closed-computation}

In this section, our goal is to present a method for calculating the optimal cost function $V^*_w$ and the stationary optimal classical-state-preserving closed-loop policy, whose existence is established in Theorem~\ref{closed-quantum-stationary}. We start with computation of the optimal cost function $V^*_w$. Recall that $V^*_w(\rho) = \langle \N_{qc}(\xi^*),\rho \rangle$, where $\xi^*$ is the operator in Assumption~\ref{as2}. Hence, to compute $V^*_w(\rho)$, it is sufficient to compute $\N_{qc}(\xi^*)$. We can accomplish this with the following method.

\begin{algorithm}[H] 
\caption{Computing the optimal cost function $V^*_w(\rho)$}
\begin{algorithmic}\label{closed-opt-cost}
\FOR{$|x\rangle \langle x| \in \S$}
\STATE{
Solve $(\textbf{SDP-w}^{\dag})$ for $\rho_0 = |x\rangle \langle x|$ and let $\xi_{|x\rangle \langle x|}$ be an optimal solution (not necessarily $\xi^*$).
}
\ENDFOR
\RETURN{$\sum_{x \in \sX} \tr(\xi_{|x\rangle \langle x|}|x\rangle \langle x|) \, |x\rangle \langle x|$
 and $V^*_w(\cdot) = \left\langle \sum_{x \in \sX} \tr(\xi_{|x\rangle \langle x|}|x\rangle \langle x|) \, |x\rangle \langle x|,\cdot \right \rangle $}
\end{algorithmic}
\end{algorithm}

\begin{proposition}\label{closed-opt-cost-prop}
The method above returns optimal cost function $V^*_w(\cdot)$.
\end{proposition}

\begin{proof}
For any $|x\rangle \langle x| \in \S$, let $\xi_{|x\rangle \langle x|}$ be an optimal solution for $(\textbf{SDP-w}^{\dag})$ for $\rho_0 = |x\rangle \langle x|$. Since $\xi^*$ in Assumption~\ref{as2} is also an optimal solution of  $(\textbf{SDP-w}^{\dag})$ regardless of the initial state, we have 
\begin{align*}
\langle \xi^*,(1-\beta) \N_{qc}(|x\rangle \langle x|) \rangle = \langle \xi_{|x\rangle \langle x|},(1-\beta) \N_{qc}(|x\rangle \langle x|) \rangle.
\end{align*}
Since $\N_{qc}(|x\rangle \langle x|) = |x\rangle \langle x|$ for any $|x\rangle \langle x| \in \S$, we have 
$$
\tr(\xi^* |x\rangle \langle x|) = \langle \xi^*,|x\rangle \langle x| \rangle = \langle \xi_{|x\rangle \langle x|},|x\rangle \langle x| \rangle = \tr(\xi_{|x\rangle \langle x|} |x\rangle \langle x| ). 
$$
This implies that 
$$
\N_{qc}(\xi^*) = \sum_{x \in \sX} \tr(\xi_{|x\rangle \langle x|}|x\rangle \langle x|) \, |x\rangle \langle x|,
$$
and so 
$$V^*_w(\cdot) = \left\langle \sum_{x \in \sX} \tr(\xi_{|x\rangle \langle x|}|x\rangle \langle x|) \, |x\rangle \langle x|,\cdot \right \rangle. $$
\end{proof}

In the procedure above, we need to solve $(\textbf{SDP-w}^{\dag})$ $\dim(\H_{\sX})$-times. Since $(\textbf{SDP-w}^{\dag})$ is a semi-definite programming problem, it can be solved efficiently by various well-known algorithms (see \cite{WoSaVa00}). Hence, this approach is computationally feasible if the dimension of $\H_{\sX}$ is not too large 

Now it is time to present a method for computing stationary optimal classical-state-preserving closed-loop policy. Note that the state-action occupation density operator of the optimal stationary classical-state-preserving closed-loop policy $\bgamma^* = \{\gamma^*\}$ can be written in the following form
$$
\sigma^{\bgamma^*} = (1-\beta) \, \sum_{t=0}^{\infty} \beta^t \, \gamma^*(\rho_t^{\bgamma^*}) = \gamma^*(\rho^{\bgamma^*}).
$$
Hence, $\sigma^{\bgamma^*}$ is an optimal solution of $(\textbf{SDP-w})$, which also satisfies the following condition
\begin{align*}
\sigma^{\bgamma^*} &= \gamma^*(\rho^{\bgamma^*}) \\
&= \gamma^*\left((1-\beta) \, \rho_0 + (1-\beta) \, \sum_{t=1}^{\infty} \beta^t \, \N(\sigma_{t-1}^{\bgamma^*}) \right) \\
&= (1-\beta) \, \gamma^*(\rho_0) + \beta \, \gamma^* \circ \N(\sigma^{\bgamma^*}) \\
&= \gamma^*\left((1-\beta) \, \rho_0 + \beta \, \N(\sigma^{\bgamma^*}) \right)
\end{align*}
The last identity suggests the following characterization of the stationary optimal classical-state-preserving closed-loop policy. 

\begin{proposition}\label{closed-optimal-policy-prop}
Let $\sigma^*$ be an optimal solution for $(\textbf{SDP-w})$, which also satisfies the additional condition: there exists $\gamma^* \in \cC_w$ such that 
$$
\gamma^*\left((1-\beta) \, \rho_0 + \beta \, \N(\sigma^{*}) \right) = \sigma^{*},
$$
that is; the reverse of the quantum channel $\D(\H_{\sX}\otimes\H_{\sA}) \ni \sigma \rightarrow (1-\beta) \, \rho_0 + \beta \, \N(\sigma) \in \D(\H_{\sX})$ is $\gamma^*$ at $\sigma^*$. Then, $\bgamma^* = \{\gamma^*\}$ is an optimal stationary classical-state-preserving closed-loop policy. 
\end{proposition}

\begin{proof}
The existence of optimal stationary classical-state-preserving closed-loop  policy implies that there is at least one such optimal solution available for $(\textbf{SDP-w})$ satisfying this additional condition. 

It is worth noting that there can be multiple optimal solutions in this form for $(\textbf{SDP-w})$, and they may not necessarily corresponds to a state-action occupation density operator of an optimal stationary classical-state-preserving closed-loop  policy. However, we will demonstrate that this is not the case. Indeed, let $\bgamma^*$ be the stationary classical-state-preserving closed-loop policy in the statement of the proposition. Then, we have 
\begin{align*}
\sigma^{*} &= \gamma^*\left((1-\beta) \, \rho_0 + \beta \, \N(\sigma^{*}) \right)  \\
&= \gamma^*\left((1-\beta) \, \rho_0 + \beta \, \N(\gamma^*\left((1-\beta) \, \rho_0 + \beta \, \N(\sigma^{*}) \right)) \right)\\
&\phantom{x}\vdots \\
&\rightarrow (1-\beta) \, \sum_{t=0}^{\infty} \beta^t \, \sigma_t^{\bgamma^*} 
\end{align*}
Hence, state-action occupation density operator under the stationary policy $\bgamma^*$ is $\sigma^{*}$. Since $\sigma^*$ is optimal for  $(\textbf{SDP-w})$, we have  
\begin{align*}
(1-\beta) \, V^*_w(\rho_0) = (\textbf{SDP-w}) = \langle c,\sigma^* \rangle = (1-\beta) \, V(\bgamma^*,\rho_0).
\end{align*} 
Hence, $\bgamma^*$ is the optimal stationary classical-state-preserving closed-loop policy. 
\end{proof}

Typically, for any initial state $\rho_0$, the semi-definite programming problem $(\textbf{SDP-w})$ can be efficiently solved using various well-known algorithms (see \cite{WoSaVa00}). However, in Proposition~\ref{closed-optimal-policy-prop}, we need to find an optimal solution to $(\textbf{SDP-w})$ that meets an additional condition. As we will discuss below, this task can be formulated as a bi-linear optimization problem, which can be computationally demanding to solve.

To derive the bi-linear optimization problem, we begin by introducing the Choi matrix representation of quantum channels. Let $\theta$ be a quantum channel from $\H_1$ to $\H_2$. Let $\{|e_i\rangle\}_{i=1}^n$ be an orthonormal basis for $\H_1$. Then, Choi matrix corresponding to the super-operator $\theta$ from $\H_1$ to $\H_2$ is defined as 
$$
\mathfrak{C_{\theta}} \coloneqq \sum_{i,j=1}^n |e_i\rangle \langle e_j| \otimes \theta(|e_i\rangle \langle e_j|).
$$ 
Therefore, $\mathfrak{C_{\theta}} \in \cL(\H_1 \otimes \H_2)$. It is known that (see \cite[Theorem 4.3]{KhWi20}) the super-operator $\theta$ is quantum channel if and only if the Choi matrix  $\mathfrak{C_{\theta}}$ is positive semi-definite (i.e. $\mathfrak{C_{\theta}} \succcurlyeq 0$) and $\tr_{\H_2}(\mathfrak{C_{\theta}}) = \Id_{\H_1}$. The action of the quantum channel $\theta$ can be described using its Choi representation $\mathfrak{C_{\theta}}$. Specifically, for any $\rho \in \H_1$, we have (see \cite[Proposition 4.2]{KhWi20}) 
$$
\theta(\rho) = \langle \psi_{\ent}| \left(\rho \otimes \mathfrak{C_{\theta}} \right) |\psi_{\ent} \rangle \eqqcolon \Lambda(\mathfrak{C_{\theta}},\rho),
$$
where $|\psi_{\ent} \rangle \coloneqq \sum_{i=1}^n |e_i\rangle \otimes |e_i\rangle$ is the unnormalized maximally entangled state. Here, $\langle \psi_{\ent}| \, \cdot \, |\psi_{\ent}\rangle$ acts as follows: $\rho_{1,1} \otimes \rho_2 \mapsto \langle \psi_{\ent}| \rho_{1,1 } |\psi_{\ent}\rangle \, \rho_2$.
Observe that $\Lambda(\mathfrak{C}_{\theta}, \rho)$ is clearly bi-linear in its arguments. Using the Choi representation of quantum channels, we can now formulate the bi-linear optimization problem that yields the optimal stationary classical-state-preserving closed-loop policy:
\begin{tcolorbox} 
[colback=white!100]
\begin{align}
(\textbf{BIL-w}) \text{                         }&\min_{\substack{\sigma \in \cL_H(\H_{\sX}\otimes\H_{\sA}) \\ \mathfrak{C} \in \cL_H(\H_{\sX}\otimes\H_{\sX}\otimes\H_{\sA})}} \text{ } \langle c,\sigma \rangle
\nonumber \\*
&\text{s.t.} \, \T_w(\sigma) = (1-\beta) \,\N_{qc}(\rho_0) \,\, \text{and} \,\, \sigma \succcurlyeq 0 \nonumber \\
&\phantom{s.t.} \tr_{\sX\times\sA}(\mathfrak{C}) = \Id_{\sX}, \mathfrak{C} \succcurlyeq 0 \label{eqqq1} \\
&\phantom{s.t.}  \tr_{\sA} \circ \Lambda(\mathfrak{C},|x\rangle \langle x|) = |x\rangle \langle x|, \,\, \forall \, |x\rangle \langle x| \in \S \label{eqqq2} \\ 
&\phantom{s.t.} \Lambda(\mathfrak{C},(1-\beta) \, \rho_0 + \beta \, \N(\sigma)) = \sigma \label{eqqq3}
\end{align} 
\end{tcolorbox} 
Here, constraint in (\ref{eqqq1}) ensures that $\mathfrak{C}$ is a Choi representation of some quantum channel from $\H_{\sX}$ to $\H_{\sX} \otimes \H_{\sA}$, constraint in (\ref{eqqq2}) ensures that the quantum channel corresponding to $\mathfrak{C}$ is in $\cC_w$, and finally, constraint in (\ref{eqqq3}) ensures that the condition in Proposition~\ref{closed-optimal-policy-prop} holds. In $(\textbf{BIL-w})$, only the last constraint is non-linear, and indeed, it is bi-linear in $(\sigma,\mathfrak{C})$. With these observations, we can conclude the following:

\begin{proposition}\label{closed-optimal-policy-prop_comp}
Let $(\sigma^*,\mathfrak{C}^*)$ be an optimal solution for $(\textbf{BIL-w})$. Then, the quantum channel 
$$
\gamma^*(\rho) \coloneqq  \Lambda(\mathfrak{C^*},\rho),
$$
is the optimal stationary classical-state-preserving closed-loop policy.
\end{proposition}

To compute an optimal stationary classical-state-preserving closed-loop policy,  we need to solve the bi-linear optimization problem $(\textbf{BIL-w})$. 
As explained in open-loop case, bi-linear optimization problems are difficult to solve because of the coupling between variables in bi-linear terms, which prevents separate optimization and makes effective convex relaxations challenging. These problems are often NP-hard and solving them typically requires specialized algorithms. Therefore, we propose the development of effective algorithms for solving these bi-linear programs ($( \textbf{BIL})$ and $( \textbf{BIL-w})$) as a future research direction.

\section{Conclusion}

In this work, we study two classes of policies for q-MDPs: open-loop and classical-state-preserving closed-loop policies. For both policy classes, we show that, through duality between dynamic programming and semi-definite programming, the optimal value function is linear and a stationary optimal policy exists. We also develop methods for computing the optimal value functions and stationary optimal policies as bi-linear programs. These results offer practical tools for solving quantum MDP problems.

\subsection{Some Future Research Directions}

To compute an optimal stationary open-loop policy and a classical-state-preserving closed-loop policy, we need to solve the bi-linear optimization problems $(\textbf{BIL})$ and $(\textbf{BIL-w})$, respectively. As mentioned earlier, these problems are difficult to tackle due to their non-convex nature, which results from the coupling between variables in the bi-linear terms. This makes it hard to optimize them separately or apply effective convex relaxations. Since these problems are often NP-hard, their complexity grows rapidly with the size of the problem, requiring specialized algorithms and significant computational effort. Therefore, the first research goal is to develop efficient algorithms for solving these bi-linear optimization problems by leveraging their unique structure.

The second area of exploration involves the potential benefits of quantum policies in the quantum version of classical MDPs. Since classical-state-preserving closed-loop policies include classical policies as a subset, we aim to explore how these policies might improve the optimal cost. This line of research relates to nonlocal games and Bell inequalities in quantum information theory \cite{Sca19}, which explore cost structures where quantum policies outperform classical ones in multi-agent decision scenarios.

Finally, we are interested in extending q-MDPs to a mean-field framework. In this version, both the state dynamics and the cost function would be influenced by the state density operator, introducing non-linear effects. A key challenge here is figuring out how to integrate these effects into the quantum channel governing state dynamics and the one-stage cost without disrupting the underlying quantum structure. This generalization offers a promising and challenging direction for future research.

\section{Appendix}\label{appendix}

\subsection{Necessary and Sufficient Conditions for Assumption~\ref{as1}}\label{nec-as1}

In this section, we derive some necessary and sufficient conditions for Assumption~\ref{as1}. To this end, let us first define the following set
$$
\C \coloneqq \{X \in \cL_H(\H_{\sX}\otimes\H_{\sA}): c \succcurlyeq X \}.
$$
Define also the following subset of $\C$
$$
\C_o \coloneqq \left\{X \in \C: \min_{\pi \in \D(\H_{\sA})} \langle c-X,\rho\otimes\pi \rangle = 0, \,\, \forall \, \rho \in \D(\H_{\sX}) \right\}.
$$
First of all, $\C_o$ is not empty as $c \in \C_o$. Moreover, if $c = Y \otimes R$, where $Y$ is positive semi-definite and $R$ is positive definite, then $X = Y \otimes \lambda_{\min} \Id$ is in $\C$. Here $\lambda_{\min}>0$ is the minimum eigenvalue of $R$. Indeed, $c \succcurlyeq X$ is obviously true, and moreover, let $|\xi\rangle$ be the eigenvector of $R$ corresponding to $\lambda_{\min}$. Then, we have
$$
\langle c-X, \rho\otimes |\xi\rangle \langle \xi| \rangle = \langle Y \otimes (R-\lambda_{\min} \Id), \rho\otimes |\xi\rangle \langle \xi| \rangle = 0.
$$
Hence $X \in \C_o$. Therefore, depending on the structure of $c$, there exist elements in $\C_o$ different than $c$. Note that $\xi^*$ in Assumption~\ref{as1} is supposed to satisfy the following condition: 
$\xi^*\otimes\Id - \beta \, \N^{\dag}(\xi^*) \in \C_o.$
However, with the present definition of $\C_o$, verifying this condition can be somewhat challenging. Therefore, we must seek an alternative representation of $\C_o$ and possibly a subset thereof.

Suppose that $M \succcurlyeq 0$. Then for any $\rho \in \D(\H_{\sX})$ and $\pi \in \D(\H_{\sA})$, we have 
\begin{align*}
\langle M,\rho \otimes \pi \rangle &= \tr(M \, (\rho \otimes \pi)) \\
&= \tr(M \, (\rho \otimes \Id) \, (\Id \otimes \pi)) \\
&= \tr(\tr_{\sX}(M \, (\rho \otimes \Id) \, (\Id \otimes \pi))) \\
&= \tr(\tr_{\sX}(M \, (\rho \otimes \Id)) \, \pi). 
\end{align*}
The last expression is $0$ for some $\pi \in \D(\H_{\sA})$ if and only if $\ker(\tr_{\sX}(M \, (\rho \otimes \Id))) \neq \emptyset$. In view of this, define 
$$
K \coloneqq \left\{M \succcurlyeq 0: \ker(\tr_{\sX}(M \, (\rho \otimes \Id))) \neq \emptyset, \,\, \forall \rho \in \D(\H_{\sX}) \right\}.
$$
For any $M \succcurlyeq 0$, let us consider its Schmidt decomposition with respect to Hermitian operators \cite[page 131]{PaLu08}: there exist $r \leq \min\{\dim(\H_{\sX})^2,\dim(\H_{\sA})^2\}$, orhonormal Hermitian operators $\{E_{\sX}^k\}_{k=1}^r \subset \cL_{H}(\H_{\sX})$, $\{E_{\sA}^k\}_{k=1}^r \subset \cL_{H}(\H_{\sA})$ with respect to Hilbert-Schmidt inner product, and positive real numbers $\{\lambda_k\}_{k=1}^r$ such that 
$
M = \sum_{k=1}^r \lambda_k \, E_{\sX}^k \otimes E_{\sA}^k.
$
If $\bigcap_{k=1}^r \ker(E_{\sA}^k) \neq \emptyset$, then $M \in K$. Hence the following set is a subset of $K$
$$
K_o \coloneqq \left\{M \succcurlyeq 0: \bigcap_{k=1}^r \ker(E_{\sA}^k) \neq \emptyset, \,\, M = \sum_{k=1}^r \lambda_k \, E_{\sX}^k \otimes E_{\sA}^k \,\, \text{is Schmidt decomposition} \right\}.
$$
With these insights, we can now derive an alternative description of $\C_o$ and identify a subset of it as follows:
$$
\C_o = \left\{X \in \C: c-X \in K\right\} \supset \left\{X \in \C: c-X \in K_o\right\} \eqqcolon \C_u.
$$
Remember that Assumption~\ref{as1} posits the existence of $\xi^*$ such that $\xi^*\otimes\Id - \beta \, \N^{\dag}(\xi^*) \in \C_o$. In other words, the intersection of the range of the Hermicity preserving super-operator $\T^{\dag}:\xi \mapsto \xi \otimes\Id - \beta \, \N^{\dag}(\xi)$ and $\C_o$ is non-empty. Note that the adjoint of $\T^{\dag}$ is the Hermicity preserving super-operator $\T: \sigma \mapsto \tr_{\sA}(\sigma)-\beta \, \N(\sigma)$ that was used in the formulation of the $(\textbf{SDP})$. Thus, $\range(\T^{\dag}) = \ker(\T)^{\perp}$. Hence, we can alternatively state Assumption~\ref{as1} as follows.

\begin{assumption}
We have $\range(\T^{\dag}) \bigcap \C_o = \ker(\T)^{\perp} \bigcap \C_o \neq \emptyset$. 
\end{assumption}

Note also that since $\C_u \subset \C_o$, the following assumption implies Assumption~\ref{as1}.

\begin{assumption}
We have $\range(\T^{\dag}) \bigcap \C_u = \ker(\T)^{\perp} \bigcap \C_u \neq \emptyset$. 
\end{assumption}

\subsection{Necessary and Sufficient Conditions for Assumption~\ref{as2}}\label{sub4sub1sec3}

In this section we derive some necessary and sufficient conditions for Assumption~\ref{as2}. To this end, recall the following sets $\C$, $\C_o$, and $\C_u$ from Section~\ref{nec-as1}.
Let us also define the following set
$$
\C_w \coloneqq \left\{X \in \C: \min_{\gamma \in \cC_w} \langle c-X,\gamma(\rho)\rangle = 0, \,\, \forall \, \rho \in \D(\H_{\sX}) \right\}.
$$
First of all, $\C_o \subset \C_w$. Secondly, Assumption~\ref{as2} states there exists $\xi^*$ such that $\N_{qc}(\xi^*)\otimes\Id - \beta \, \N^{\dag} \circ \N_{qc}(\xi^*) \in \C_w$. In other words, the intersection of the range of the Hermicity preserving super-operator $\T^{\dag}_w:\xi \mapsto \N_{qc}(\xi) \otimes\Id - \beta \, \N^{\dag}\circ \N_{qc}(\xi)$ and $\C_w$ is non-empty. Note that the adjoint of $\T^{\dag}_w$ is the Hermicity preserving super-operator $\T_w: \sigma \mapsto \N_{qc} \circ \tr_{\sA}(\sigma)-\beta \, \N_{qc} \circ \N(\sigma)$. Thus $\range(\T^{\dag}_w) = \ker(\T_w)^{\perp}$. Hence, we can alternatively state Assumption~\ref{as2} as follows.

\begin{assumption}
We have $\range(\T^{\dag}_w) \bigcap \C_w = \ker(\T_w)^{\perp} \bigcap \C_w \neq \emptyset$. 
\end{assumption}

Note also that since $\C_u \subset \C_o \subset \C_w$, the following assumptions imply Assumption~\ref{as2}.

\begin{assumption}
We have $\range(\T^{\dag}_w) \bigcap \C_o = \ker(\T_w)^{\perp} \bigcap \C_o \neq \emptyset$. 
\end{assumption}

\begin{assumption}
We have $\range(\T^{\dag}_w) \bigcap \C_u = \ker(\T_w)^{\perp} \bigcap \C_u \neq \emptyset$. 
\end{assumption}

\subsection{Revisiting Classical MDPs: Deterministic Reduction and LP Formulation}\label{sec1}

We revisit the exposition of classical MDPs as briefly mentioned in Section~\ref{sec2new-1}. Here, we gather well-known results concerning MDPs, including deterministic reduction, dynamic programming, and the linear programming formulation of deterministic reduction. Although these findings might not present novel insights, their setup plays a crucial role in enabling a parallel examination within the quantum domain. While similar analyses apply to finite horizon cost and average cost criteria, we focus on the discounted cost structure.

Recall that the deterministic reduction of a MDP, denoted as d-MDP, can be described by a tuple 
$$
\left( \P(\sX), \P(\sX \times \sA),  \{\sA(\mu): \mu \in \P(\sX)\}, P, C \right),
$$ 
where $\P(\sX)$ is the state space and $\P(\sX \times \sA)$ is the action space, each endowed with weak convergence topology. The state transition function $P: \P(\sX \times \sA) \rightarrow \P(\sX)$ is given by
\begin{align}
P(\nu)(\,\cdot\,) = \sum_{(x,a) \in \sX \times \sA} p(\,\cdot\,|x,a) \, \nu(x,a). \label{eq1}
\end{align}
The \emph{one-stage cost} function $C$ is a linear function from $\P(\sX \times \sA)$ to $\R$ and is given by
\begin{align}
C(\nu) = \sum_{(x,a) \in \sX \times \sA} c(x,a) \, \nu(x,a) \eqqcolon \langle c,\nu \rangle. \label{eq2}
\end{align}
In this model, a policy is a sequence of classical channels $\bgamma=\{\gamma_{t}\}$ from $\sX$ to $\sX\times\sA$ where for all $t$
\begin{align*} 
\gamma_t(\mu) \in \sA(\mu), \,\, \forall \,\mu \in \P(\sX). 
\end{align*}
The set of all policies is denoted by $\Gamma$. A \emph{stationary} policy is a constant sequence of channels $\bgamma=\{\gamma\}$ from $\sX$ to $\sX\times\sA$, where $\gamma(\mu) \in \sA(\mu)$ for all $\mu \in \P(\sX)$.

Using above definition for policies in d-MDPs, we can establish that for any  $\gamma_t:\P(\sX) \rightarrow \P(\sX\times\sA)$ satisfying $\gamma_t(\mu) \in \sA(\mu)$ for all $\mu \in \P(\sX)$, there exists a stochastic kernel $\pi_t$ from $\sX$ to $\sA$ such that $\gamma_t(\mu) = \mu \otimes \pi_t$, where $\mu \otimes \pi_t(x,a) \coloneqq \mu(x) \, \pi_t(a|x)$ (see Section 2.3 in \cite{SaSaYu24}). Let us write this relation as $\pi_t \coloneqq \Lambda(\gamma_t)$. This gives a bijective relation between the set of policies $\Pi$ of the original MDP and the set of policies $\Gamma$ of the d-MDP. By an abuse of notation, let us also denote this bijective relation by $\Lambda:\Pi \rightarrow \Gamma$. This bijective relation indeed establishes the equivalence of the original MDP and the deterministic reduction d-MDP.

\subsubsection{Dynamic Programming for d-MDP}\label{sub2sec1}

In this section, we introduce the dynamic programming equation for obtaining the optimal value function and policy, a well-established result in the literature (see \cite{BeSh78,HeLa96}).

The dynamic programming operator that gives the optimal value function as well as the optimal policy for d-MDP is the operator $\L:C_b(\P(\sX)) \rightarrow C_b(\P(\sX))$, which is defined as
$$
\L V(\mu) \coloneqq \min_{\nu \in \sA(\mu)} \left[ \langle c,\nu \rangle + \beta \, V(P(\nu)) \right].
$$
Note that we can alternatively define $\L$ using stochastic kernels as follows
$$
\L V(\mu) \coloneqq \min_{\pi \in \P(\sA|\sX)} \left[ \langle c,\mu \otimes \pi \rangle + \beta \, V(P(\mu \otimes \pi )) \right],
$$
where $\P(\sA|\sX)$ denotes the set of stochastic kernels from $\sX$ to $\sA$. From classical MDP theory, it is known that the operator $\L$ is $\beta$-contraction with respect sup-norm on $C_b(\P(\sX))$, and so, has a unique fixed point $V_{\mathrm{f}} \in C_b(\P(\sX))$. In the following, we establish a connection between optimal value function $V^*$ of d-MDP and $V_{\mathrm{f}}$. 

\begin{theorem}\label{dyn-eq-dmdp}
We have $V^*(\mu_0) = V_{\mathrm{f}}(\mu_0)$ for all $\mu_0 \in \P(\sX)$; that is, the unique fixed point of $\L$ is $V^*$.
\end{theorem}

\begin{proof}
Fix any policy $\bgamma = \{\gamma_t\}_{t\geq0}$ with the corresponding stochastic kernels $\{\pi_t\}_{t\geq0}$; that is, $\pi_t = \Lambda(\gamma_t)$. Then, the cost function of $\bgamma$ satisfies the following 
\begin{align*}
V(\bgamma,\mu_0) &= \langle c,\mu_0 \otimes \pi_0 \rangle + \beta \, V(\{\gamma_t\}_{t\geq 1},P(\mu_0\otimes\pi_0)) \\
&\geq \min_{\pi_0 \in \P(\sA|\sX)} \left[ \langle c,\mu_0 \otimes \pi_0 \rangle + \beta \, V(\{\gamma_t\}_{t\geq 1},P(\mu_0\otimes\pi_0)) \right] \\
&\geq \min_{\pi_0 \in \P(\sA|\sX)} \left[ \langle c,\mu_0 \otimes \pi_0 \rangle + \beta \, V^*(P(\mu_0\otimes\pi_0)) \right] = \L V^*(\mu_0).
\end{align*}
Since above inequality is true for any policy $\bgamma$, we have $V^* \geq \L V^*$. Using the latter inequality and the monotonicity of $\L$, one can also prove that $V^* \geq \L^nV^*$ for any positive integer $n$. By Banach fixed point theorem, it is known that $\L^n V^* \rightarrow V_{\mathrm{f}}$ in sup-norm as $n\rightarrow \infty$. Hence, $V^* \geq  V_{\mathrm{f}}$. 

To prove the converse inequality, for any $\mu$, let 
$$
\pi_{\mu} \in \argmin_{\pi \in \P(\sA|\sX)} \left[ \langle c,\mu \otimes \pi \rangle + \beta \, V_{\mathrm{f}}(P(\mu \otimes \pi )) \right].
$$
The existence of $\pi_{\mu}$ for any $\mu$ follows from the facts that $\sA(\mu)$ is compact and  $\langle c,\nu  \rangle + \beta \, V_{\mathrm{f}}(P(\nu))$ is continuous in $\nu$. We then have
\begin{align*}
V_{\mathrm{f}}(\mu_0) &=  \langle c,\mu_0 \otimes \pi_{\mu_0} \rangle + \beta \, V_{\mathrm{f}}(P(\mu_0 \otimes \pi_{\mu_0} )) \\
&\eqqcolon \langle c,\nu_0  \rangle + \beta \, V_{\mathrm{f}}(\mu_1) \\
&= \langle c,\nu_0  \rangle + \beta \, \left[\langle c,\mu_1 \otimes \pi_{\mu_1} \rangle + \beta \, V_{\mathrm{f}}(P(\mu_1 \otimes \pi_{\mu_1} )) \right] \\
&\eqqcolon \langle c,\nu_0  \rangle + \beta \, \left[\langle c,\nu_1 \rangle + \beta \, V_{\mathrm{f}}(\mu_2) \right] \\
&\phantom{x}\vdots \\
&= \sum_{t=0}^{N-1} \beta^t \, \langle c,\nu_t \rangle + \beta^N \, V_{\mathrm{f}}(\mu_N) \rightarrow V(\{\gamma_t\}_{t\geq0},\mu_0) \,\, \text{as} \,\ N\rightarrow\infty,
\end{align*}
where $\gamma_t(\mu) \coloneqq \mu \otimes \pi_{\mu_t}$ for each $t$. Hence, $V_{\mathrm{f}} \geq V^*$. This completes the proof.
\end{proof}

Upon analyzing the latter part of the above proof, it becomes apparent that the optimal policy $\bgamma^*$ for any given initial point $\mu_0$ can be constructed recursively in the following manner: $\gamma^*_0(\mu) = \mu \otimes \pi_{\mu_0}$ and $\gamma^*_{t+1}(\mu) = \mu \otimes \pi_{\mu_{t+1}}$ for $t\geq0$, where $\mu_{t+1} = P(\mu_t \otimes \pi_{\mu_t})$. As a result, in order to determine the optimal policy for $\mu_0$, it is necessary to know the optimal value function $V^*$. Furthermore, it appears that the optimal policy is contingent upon the initial distribution and varies with time (non-stationary). However, in the subsequent section, by utilizing a linear programming formulation of d-MDP, it is demonstrated that the optimal value function $V^*(\mu)$ is linear with respect to $\mu$, i.e., there exists a vector $\xi^* \in \R^{\sX}$ such that $V^*(\mu) = \langle \xi^*,\mu \rangle$. Subsequently, it is proven that an optimal stationary policy exists\footnote{Here, we are cheating a little bit because the linearity of the value function $V^*$ can also be demonstrated using the dynamic programming equation. Consequently, establishing the stationarity  of the optimal policy is straightforward by considering Dirac-delta measures first. However, we intentionally follow this approach to draw an analogy with the quantum case. In quantum scenarios, it is not possible to establish these facts solely through the dynamic programming equation.}.

\subsubsection{Linear Programming Formulation of d-MDP}\label{sub3sec1}

In this section, we derive the linear programming formulation for d-MDPs, which provides the optimal value function and policy. Our approach closely mirrors the method employed to derive the linear programming formulation for classical stochastic MDPs, as in \cite[Chapter 6]{HeLa96}. 

Consider the d-MDP 
\begin{align}
\left( \P(\sX), \P(\sX \times \sA),  \{\sA(\mu): \mu \in \P(\sX)\}, P, C \right) \nonumber
\end{align}
under $\beta$-discounted cost criterion with initial point $\mu_0$. Given any policy $\bgamma$, define the state-action occupation probability measure as follows
$$
\nu^{\bgamma} \coloneqq (1-\beta) \, \sum_{t=0}^{\infty} \beta^t \, \nu_t,
$$
where the term $(1-\beta)$ is the normalization constant. Then, one can write the normalized cost of $\bgamma$ in the following form
$$
(1-\beta) \, V(\bgamma,\mu_0) =  C(\nu^{\bgamma}) = \langle c,\nu^{\bgamma} \rangle
$$
as a result of linearity of $C$. Let $\mu^{\bgamma} \coloneqq \nu^{\bgamma}(\,\cdot\,\times \sA)$. We then have 
\begin{align}
\mu^{\bgamma}(\,\cdot\,) &= (1-\beta) \, \sum_{t=0}^{\infty} \beta^t \, \nu_t(\,\cdot\, \times \sA) \nonumber \\
&= (1-\beta) \, \sum_{t=0}^{\infty} \beta^t \, \gamma_t(\mu_t)(\,\cdot\, \times \sA) \nonumber \\
&= (1-\beta) \, \sum_{t=0}^{\infty} \beta^t \mu_t \,\, \text{(as $\gamma_t(\mu_t) \in \sA(\mu_t)$)} \nonumber \\
&= (1-\beta) \,\mu_0 + (1-\beta) \, \beta \sum_{t=1}^{\infty} \beta^{t-1} \, P(\nu_{t-1}) \nonumber \\
&= (1-\beta) \, \mu_0 + \beta \, P\left((1-\beta) \,\sum_{t=1}^{\infty} \beta^{t-1} \, \nu_{t-1}\right)  \,\, \text{(by linearity of $P$)}\nonumber \\
&= (1-\beta) \,\mu_0 + \beta \, P(\nu^{\bgamma}). \nonumber 
\end{align}
This computation suggests the definition of the linear operator $\T: \M(\sX\times\sA) \rightarrow \M(\sX)$ as follows
$$
\T(\nu) \coloneqq \mu - \beta \, P(\nu),
$$
where $\mu(\,\cdot\,) \coloneqq \nu(\,\cdot\, \times \sA)$ and $\M(\sE)$ denotes the set of finite signed measures on the set $\sE$. Therefore, $\T(\nu^{\bgamma}) = (1-\beta) \,\mu_0$ for any policy $\bgamma$. This observation leads to the following linear programming formulation of the d-MDP
\begin{align}
(\textbf{LP}) \text{                         }&\min_{\nu \in \M(\sX\times\sA)} \text{ } \langle c,\nu \rangle
\nonumber \\*
&\text{s.t.} \, \T(\nu) = (1-\beta) \,\mu_0 \,\, \text{and} \,\, \nu \geq 0. \nonumber 
\end{align} 
Note that if $\T(\nu) = (1-\beta) \,\mu_0$ and $\nu \geq 0$, then necessarily $\nu \in \P(\sX\times\sA)$. As a result, it is unnecessary to impose a prior restriction on $\nu$ to be confined within $\P(\sX\times\sA)$ in the linear programming formulation. The following theorem establishes the equivalence of $(\textbf{LP})$ and d-MDP. 

\begin{theorem}\label{lpformulation}
For any policy $\bgamma$, there exists a $\nu^{\bgamma} \in \M(\sX\times\sA)$ that is feasible for $(\textbf{LP})$ and $(1-\beta) \, V(\bgamma,\mu_0) = \langle c,\nu^{\bgamma} \rangle$. Conversely, if $\nu \in \M(\sX\times\sA)$ is feasible for $(\textbf{LP})$, then there exists a policy $\bgamma^{\nu}$ such that $(1-\beta) \, V(\bgamma^{\nu},\mu_0) = \langle c,\nu \rangle$. 
\end{theorem}

\begin{proof}
Given any policy $\bgamma$, consider the state-action occupation probability measure $\nu^{\bgamma}$. In view of discussion and computation above,  $\nu^{\bgamma}$ is feasible for $(\textbf{LP})$ and $(1-\beta) \,V(\bgamma,\mu_0) = \langle c,\nu^{\bgamma} \rangle$, which completes the proof of the first part.

Conversely, let $\nu \in \M(\sX\times\sA)$ be feasible for $(\textbf{LP})$. Disintegrate $\nu$ as follows
$$
\nu(x,a) = \mu(x) \, \pi(a|x), 
$$
where $\pi$ is a stochastic kernel from $\sX$ to $\sA$. Consider the stationary policy $\bgamma =\{\gamma\}$, where $\gamma: \P(\sX) \rightarrow \P(\sX\times\sA)$ is given by
$$
\gamma(\xi)(x,a) = \sum_{y \in \sX} \pi(a|x) \, \delta_y(x) \, \xi(y) = \xi(x) \, \pi(a|x), \,\, \xi \in \P(\sX). 
$$
Hence $\gamma(\mu) = \nu$. Since $\T(\nu) = (1-\beta) \, \mu_0$, for any $N\geq 1$, we have 
\begin{align}
\mu &= (1-\beta) \, \mu_0 + \beta \, P(\nu) \nonumber \\
&= (1-\beta) \, \mu_0 + \beta \, P \circ \gamma(\mu) \nonumber \\
&= (1-\beta) \, \mu_0 + \beta \, P \circ \gamma((1-\beta) \,\mu_0 + \beta \, P(\nu)) \nonumber \\
&= (1-\beta) \,\mu_0 + (1-\beta) \, \beta \, P \circ \gamma(\mu_0) + \beta^2 \, P \circ \gamma \circ P(\nu) \nonumber \\
&\eqqcolon (1-\beta) \,\mu_0 + (1-\beta) \, \beta \, \mu_1 + \beta^2 \, P \circ \gamma \circ P \circ \gamma(\mu) \,\, \text{(by linearity of $\gamma$ and $P$)} \nonumber \\
&\phantom{x}\vdots \nonumber \\
&= (1-\beta) \, \sum_{t=0}^{N-1} \beta^t \, \mu_t + \beta^N \, \left(P\circ\gamma\right)^N(\mu).\nonumber 
\end{align}
Then, as $N \rightarrow \infty$, we have
$$
\mu = (1-\beta) \, \sum_{t=0}^{\infty} \beta^t \, \mu_t,
$$
where $\mu_t$ is the state at time $t$ for d-MDP under the stationary policy $\bgamma$.
Now, consider the state-action occupation probability measure $\nu^{\bgamma}$ induced by $\bgamma$. Note that 
\begin{align}
\nu^{\bgamma} &\coloneqq  (1-\beta) \,\sum_{t=0}^{\infty} \beta^t \, \nu_t \nonumber \\
&= (1-\beta) \, \sum_{t=0}^{\infty} \beta^t \, \gamma(\mu_t)  \nonumber \\
&= \gamma\left((1-\beta) \, \sum_{t=0}^{\infty} \beta^t \, \mu_t\right) \,\, \text{(by linearity of $\gamma$)} \nonumber \\
&= \gamma(\mu) = \nu. \nonumber 
\end{align}
Therefore, $(1-\beta) \, V(\bgamma,\mu_0) = \langle c,\nu^{\gamma} \rangle = \langle c,\nu \rangle$. This completes the proof of the second part.
\end{proof}

Above theorem establishes that 
\begin{align}
\inf_{\bgamma \in \Gamma} (1-\beta) \, V(\bgamma,\mu_0) \equiv (\textbf{LP}) \text{                         }&\min_{\nu \in \M(\sX\times\sA)} \text{ } \langle c,\nu \rangle
\nonumber \\*
&\text{s.t.} \, \T(\nu) = (1-\beta) \,\mu_0 \,\, \text{and} \,\, \nu \geq 0. \nonumber 
\end{align} 
We note that when we analyze the latter part of the proof of the theorem above, it becomes evident that we can, without losing generality, limit our search for an optimal policy to stationary policies. However, in the following part of this section, in order to establish a connection with the quantum counterpart of the problem, we demonstrate the same result using a different method. To be precise, we first illustrate that the optimal value function $V^*$ is linear with respect to $\mu$, and based on this observation, we establish the existence of an optimal stationary policy. Hence, the subsequent part of this section is not essential for the analysis of d-MDPs; rather, it is included here to establish a link with the quantum version of the problem.

First of all, let us obtain the dual of the above LP. Indeed, we have
\begin{align*} 
(\textbf{LP})
& = \min_{\nu \in \M(\sX\times\sA)} \text{ } \langle c,\nu \rangle
\,\, \text{s.t.} \, \T(\nu) = (1-\beta) \,\mu_0 \,\, \text{and} \,\, \nu \geq 0 \nonumber \\
&=\min_{\nu \geq 0} \, \max_{\xi \in \R^{\sX}} \, \langle c,\nu \rangle  + \langle \xi, (1-\beta) \,\mu_0 -\T(\nu) \rangle \\
&\text{(by Sion min-max theorem \cite{Sio58})} \\
&= \max_{\xi \in \R^{\sX}} \, \min_{\nu \geq 0} \, \langle c,\nu \rangle  + \langle \xi, (1-\beta) \,\mu_0 -\T(\nu) \rangle \\
&= \max_{\xi \in \R^{\sX}} \, \min_{\nu \geq 0} \, \langle \xi,(1-\beta) \, \mu_0 \rangle  + \langle c-{\T}^{\dag}(\xi), \nu \rangle \\
&=\max_{\xi \in \R^{\sX}} \text{ }  \langle \xi,(1-\beta) \, \mu_0 \rangle 
\,\, \text{s.t.} \, c-{\T}^{\dag}(\xi) \geq  0 \nonumber \\
&\eqqcolon (\textbf{LP}^{\dag}), 
\end{align*} 
where ${\T}^{\dag}:\R^{\sX} \rightarrow \R^{\sX\times\sA}$ is the adjoint of $\T$ and is given by
$$
{\T}^{\dag}(\xi)(x,a) = \xi(x) - \beta \, P^{\dag}(\xi)(x,a) \coloneqq \xi(x) - \beta \, \sum_{y \in \sX} \xi(y) \, p(y|x,a).
$$
For the dual LP, we have the following equivalent formulations\footnote{In the quantum version of the problem, regrettably, we cannot demonstrate the equivalence of these formulations for the dual semi-definite program. Consequently, we utilize an alternative approach to establish the existence of an optimal stationary policy in the quantum case.}
\begin{align*}
&(\textbf{LP}^{\dag}) \\
&=\max_{\xi \in \R^{\sX}} \text{ }  \langle \xi,(1-\beta) \, \mu_0 \rangle 
\,\, \text{s.t.}  \, c + \beta \, P^{\dag}(\xi) \geq \xi \\
&=\max_{\xi \in \R^{\sX}} \text{ }  \langle \xi,(1-\beta) \, \mu_0 \rangle 
\,\, \text{s.t.}  \, \langle c + \beta \, P^{\dag}(\xi), \delta_{x,a} \rangle \geq \langle \xi,\delta_{x,a} \rangle \,\, \forall (x,a) \in \sX\times\sA \\
&=\max_{\xi \in \R^{\sX}} \text{ }  \langle \xi,(1-\beta) \, \mu_0 \rangle 
\,\, \text{s.t.}  \, \langle c + \beta \, P^{\dag}(\xi), \delta_{x} \otimes \pi \rangle \geq \langle \xi,\delta_{x} \otimes \pi \rangle \,\, \forall x \in \sX, \,\, \forall \pi \in \P(\sA) \\
&=\max_{\xi \in \R^{\sX}} \text{ }  \langle \xi,(1-\beta) \, \mu_0 \rangle 
\,\, \text{s.t.}  \, \langle c + \beta \, P^{\dag}(\xi), \delta_{x} \otimes \pi \rangle \geq \langle \xi,\delta_{x} \otimes \pi \rangle \,\, \forall x \in \sX, \,\, \forall \pi \in \P(\sA|\sX) \\
&=\max_{\xi \in \R^{\sX}} \text{ }  \langle \xi,(1-\beta) \, \mu_0 \rangle 
\,\, \text{s.t.}  \, \langle c + \beta \, P^{\dag}(\xi), \mu \otimes \pi \rangle \geq \langle \xi,\mu \otimes \pi \rangle \,\, \forall \mu \in \P(\sX), \,\, \forall \pi \in \P(\sA|\sX)\\ 
&=\max_{\xi \in \R^{\sX}} \text{ }  \langle \xi,(1-\beta) \, \mu_0 \rangle 
\,\, \text{s.t.}  \, \langle c + \beta \, P^{\dag}(\xi), \nu \rangle \geq \langle \xi,\nu \rangle \,\, \forall \nu \in \P(\sX\times\sA).
\end{align*}  
Now for any vector $\xi \in \R^{\sX}$, define $\L_d\xi \in \R^{\sX}$ as 
\begin{align*}
\L_d\xi(x) &\coloneqq \min_{\pi \in \P(\sA|\sX)} \langle c+ \beta \, P^{\dag}(\xi), \delta_{x} \otimes \pi \rangle \\
&=\min_{\pi \in \P(\sA|\sX)} \left[ \langle c,\delta_x \otimes \pi \rangle + \beta \, \langle \xi,P(\delta_x \otimes \pi ) \rangle \right] \\
&= \min_{\pi \in \P(\sA)} \left[ \langle c,\delta_x \otimes \pi \rangle + \beta \, \langle \xi,P(\delta_x \otimes \pi ) \rangle \right]  \\
&=\min_{a \in \sA} \left[ \langle c,\delta_{x,a} \rangle + \beta \, \langle \xi,P(\delta_{x,a}) \rangle \right] \\
&=\min_{a \in \sA} \left[ c(x,a) + \beta \, \sum_{y \in \sX} \xi(y) \, p(y|x,a) \right].
\end{align*}
The final expression suggests that $\L_d$ is indeed the dynamic programming operator for the original MDP. Hence, it is known that $\L_d$ is $\beta$-contraction under sup-norm on $\R^{\sX}$ with a unique fixed point $\xi^* \in \R^{\sX}$. Moreover, since $\L_d$ is also monotone, if $\L_d \xi \geq \xi$, then $\xi^* \geq \xi$. Given these observations, suppose $\xi$ is a feasible point for $(\textbf{LP}^{\dag})$, meaning that $\langle c + \beta \, P^{\dag}(\xi), \delta_{x} \otimes \pi \rangle \geq \langle \xi,\delta_{x} \otimes \pi \rangle$ for all $x \in \sX$ and all $\pi \in \P(\sA|\sX)$. Consequently, $\L_d \xi \geq \xi$, which implies that $\xi^* \geq \xi$. Furthermore, it is evident that $\xi^*$ is also a feasible point for $(\textbf{LP}^{\dag})$. As a result, $\xi^*$ is the optimal solution for $(\textbf{LP}^{\dag})$. It should be noted that $\xi^*$ is indeed the optimal solution for $(\textbf{LP}^{\dag})$ regardless of the initial point $\mu_0$, meaning that for any $\mu_0$, $\xi^*$ is a common optimal solution for $(\textbf{LP}^{\dag})$. Since there is no duality gap between $(\textbf{LP}^{\dag})$ and $(\textbf{LP})$ for any $\mu_0$, it follows that $(1-\beta) \, V^*(\mu_0) = (1-\beta) \, \langle \xi^*,\mu_0 \rangle$ for any $\mu_0$. This establishes that the optimal value function of d-MDP $V^*(\mu) = \langle \xi^*,\mu \rangle$ is linear.

Note that for any $\mu \in \P(\sX)$, since $V^* = \L V^*$ and $\xi^* = \L_d \xi^*$, we have 
\begin{align*}
\langle \xi^*,\mu \rangle &= \min_{\pi \in \P(\sA|\sX)}  \left[ \langle c,\mu \otimes \pi \rangle + \beta \, \langle \xi^*,P(\mu \otimes \pi ) \rangle \right]  \eqqcolon \L \langle \xi^*,\mu \rangle \\
&= \sum_{x \in \sX} \xi^*(x) \, \mu(x) \\
&= \sum_{x \in \sX} \L_d\xi^*(x) \, \mu(x) \\
&= \sum_{x \in \sX} \min_{\pi \in \P(\sA)} \left[ \langle c,\delta_x \otimes \pi \rangle + \beta \, \langle \xi^*,P(\delta_x \otimes \pi ) \rangle \right]  \, \mu(x). 
\end{align*}
Above equalities imply that if for each $x \in \sX$ 
$$\pi_x^* \in \argmin_{\pi \in \P(\sA)} \left[ \langle c,\delta_x \otimes \pi \rangle + \beta \, \langle \xi^*,P(\delta_x \otimes \pi ) \rangle \right],$$ 
then the stochastic kernel 
$$\{\pi^*(\,\cdot\,|x)\}_{x \in \sX} \coloneqq \{\pi_x^*\}_{x \in \sX} \in \argmin_{\pi \in \P(\sA|\sX)}  \left[ \langle c,\mu \otimes \pi \rangle + \beta \, \langle \xi^*,P(\mu \otimes \pi ) \rangle \right]$$ 
Since $\mu$ is arbitrary, we have 
$$\pi^* = \pi_{\mu} \coloneqq \argmin_{\pi \in \P(\sA|\sX)}  \left[ \langle c,\mu \otimes \pi \rangle + \beta \, \langle \xi^*,P(\mu \otimes \pi ) \rangle \right]$$ 
for any $\mu \in \P(\sX)$. In other words, regardless of the value of $\mu$, $\pi^*$ is the common minimizer of the dynamic programming equation. Hence the stationary policy $\bgamma^* = \{\gamma^*\}$ with $\gamma^*(\mu) = \mu \otimes \pi^*$ is optimal by dynamic programming principle.


\end{document}